\documentclass[journal, a4paper]{IEEEtran}

\usepackage[dvipsnames]{xcolor}
\usepackage[hidelinks]{hyperref}
\definecolor{TUMBlue}{RGB}{0,101,189} %
\definecolor{TUMBlueDark}{RGB}{0,82,147} %
\definecolor{TUMBlueLight}{RGB}{152,198,234} %
\definecolor{TUMBlueMedium}{RGB}{100,160,200} %

\definecolor{TUMOrange}{RGB}{227,114,34} %
\definecolor{TUMIvory}{RGB}{218,215,203} %
\definecolor{TUMGreen}{RGB}{162,173,0} %
\definecolor{TUMGray}{gray}{0.6} %
\definecolor{TUMGrayDark}{gray}{0.3} %

\definecolor{TUMGreenDark}{RGB}{0,124,48} %
\definecolor{TUMRed}{RGB}{196,7,27} %

\definecolor{plotColor1}{RGB}{0,101,189} %
\definecolor{plotColor2}{RGB}{0,124,48} %
\definecolor{plotColor3}{RGB}{196,7,27} %
\definecolor{plotColor4}{RGB}{227,114,34} %
\definecolor{plotColor5}{RGB}{0,82,147} %
\definecolor{plotColor6}{RGB}{162,173,0} %
\definecolor{plotColor7}{gray}{0.3} %

\definecolor{plotColorIA1}{RGB}{0,124,48} %
\definecolor{plotColorIA2}{RGB}{0,82,147} %
\definecolor{plotColorIA3}{RGB}{100,160,200} %
\definecolor{plotColorIA4}{RGB}{152,198,234} %
\definecolor{plotColorIA5}{gray}{0.3} %
\definecolor{plotColorIA6}{RGB}{227,114,34} %
\definecolor{plotColorIA7}{RGB}{196,7,27} %

\usepackage{amsfonts, amsmath, amssymb, amsthm}
\usepackage{bbm,cite, mathtools}
\usepackage{comment,enumitem}
\newtheorem{theorem}{Theorem}
\newtheorem*{theorem*}{Theorem}
\newtheorem{proposition}{Proposition}
\newtheorem{lemma}{Lemma}

\newtheorem{corollary}{Corollary}

\theoremstyle{definition}
\newtheorem{defn}{Definition}
\theoremstyle{remark}
\newtheorem*{remark}{Remark}
\newtheorem*{example}{Example}
\usepackage{tikz}
\usepackage{pgfplots}
\pgfplotsset{compat=newest}
\usepackage[
font=small,labelfont=bf
]{caption}

\newif\ifcomment
\commenttrue

\newif\ifhighlightChanges
\highlightChangestrue

\newcommand{\eqdef}{\vcentcolon=}

\newcommand{\x}{\ensuremath{\mathbf{x}}}
\newcommand{\X}{\ensuremath{\mathbf{X}}}
\newcommand{\Y}{\ensuremath{\mathbf{Y}}}
\renewcommand{\d}{\ensuremath{\mathbf{d}}}
\renewcommand{\i}{\ensuremath{\mathbf{i}}}
\renewcommand{\j}{\ensuremath{\mathbf{j}}}
\newcommand{\y}{\ensuremath{\mathbf{y}}}

\renewcommand{\c}{\ensuremath{\mathbf{c}}}
\newcommand{\w}{\ensuremath{\mathbf{w}}}
\renewcommand{\v}{\ensuremath{\mathbf{v}}}

\renewcommand{\a}{\ensuremath{\mathbf{a}}}

\newcommand{\0}{\ensuremath{\mathbf{0}}}

\newcommand{\F}{\ensuremath{\mathbb{F}}}
\newcommand{\cF}{\ensuremath{\mathcal{F}}}
\newcommand{\C}{\ensuremath{\mathcal{C}}}
\newcommand{\A}{\ensuremath{\mathfrak{A}}}
\renewcommand{\L}{\ensuremath{\mathcal{L}}}
\newcommand{\Z}{\ensuremath{\mathbb{Z}}}
\renewcommand{\epsilon}{\varepsilon}

\newcommand{\Mod}[1]{\ (\mathrm{mod}\ #1)}
\newcommand{\Mods}[1]{\ (\mathrm{mod}^*\ #1)}
\newcommand{\Modsp}[2]{\ (\mathrm{mod}_{#2}^*\ #1)}
\newcommand{\Flead}{F_{\mathrm{lead}}}
\newcommand{\Fdrop}{F_{\mathrm{drop}}}
\newcommand{\rB}{r^{\mathrm{B}}}
\newcommand{\NB}{N^{\mathrm{B}}}
\newcommand{\NP}{N^{\mathrm{P}}}
\newcommand{\rP}{r^{\mathrm{P}}}
\DeclareMathOperator{\pair}{Pair}

\usepackage{array}
\usepackage{arydshln}
\newcolumntype{C}{>{$}c<{$}} %

\begin{document}
\title{Lifted Reed-Solomon Codes\\ and Lifted Multiplicity Codes}
\author{%
    \IEEEauthorblockN{Lukas~Holzbaur, Rina~Polyanskaya, Nikita~Polyanskii, Ilya~Vorobyev, and Eitan~Yaakobi}\\
	\thanks{
      The results on lifted RS codes have partially been presented at the IEEE International Symposium on Information Theory (ISIT) 2020~\cite{holzbaur2020lifted} and parts of the results on lifted multiplicity codes have been presented at the IEEE Information Theory Workshop (ITW) 2020~\cite{holzbaur2021lifted}.

      L. Holzbaur's work was supported by the Technical University of Munich -- Institute for Advanced Study, funded by the German Excellence Initiative and European Union 7th Framework Programme under Grant Agreement No. 291763 and the German Research Foundation (Deutsche Forschungsgemeinschaft, DFG) under Grant No. WA3907/1-1. N. Polyanskii's work was supported by a grant from the Russian Science Foundation (grant no. 19-71-00137).    E. Yaakobi's work was supported in part by the Israel Science Foundation  under Grant No. 1817/18 and by the Technion Hiroshi Fujiwara Cyber Security Research Center and the Israel National Cyber Directorate.

      L.~Holzbaur is with the Institute for Communications Engineering, Technical University of Munich, Germany. R.~Polyanskaya is with the Institute for Information Transmission Problems, Russian Academy of Sciences, Russia. N.~Polyanskii is with the  Center for Computational and Data-Intensive Science and Engineering, Skolkovo Institute of Science and Technology, Russia, and the Institute for Communications Engineering, Technical University of Munich, Germany. I.~Vorobyev is with the  Center for Computational and Data-Intensive Science and Engineering, Skolkovo Institute of Science and Technology, Russia. E.~Yaakobi is with the Computer Science Department, Technion --- Israel Institute of Technology, Israel.

Emails: lukas.holzbaur@tum.de, rina.polianskaia@gmail.com, nikita.polyansky@gmail.com, vorobyev.i.v@yandex.ru, yaakobi@cs.technion.ac.il
    }}

  \IEEEoverridecommandlockouts
  \maketitle

\begin{abstract}

Lifted Reed-Solomon and multiplicity codes are classes of codes, constructed from specific sets of $m$-variate polynomials. These codes allow for the design of high-rate codes that can recover every codeword or information symbol from many disjoint sets. Recently, the underlying approaches have been combined for the bi-variate case to construct lifted multiplicity codes, a generalization of lifted codes that can offer further rate improvements. We continue the study of these codes by first establishing new lower bounds on the rate of lifted Reed-Solomon codes for any number of variables $m$, which improve upon the known bounds for any $m\ge 4$. Next, we use these results to provide lower bounds on the rate and distance of lifted multiplicity codes obtained from polynomials in an arbitrary number of variables, which improve upon the known results for any $m\ge 3$.
 Specifically, we investigate a subcode of a lifted multiplicity code formed by the linear span of $m$-variate monomials whose restriction to an arbitrary line in $\F_q^m$ is equivalent to a low-degree univariate polynomial. We find the tight asymptotic behavior of the fraction of such monomials when the number of variables $m$ is fixed and the alphabet size $q=2^\ell$ is large.

Using these results, we give a new explicit construction of batch codes utilizing lifted Reed-Solomon codes. For some parameter regimes, these codes have a better trade-off between parameters than previously known batch codes. Further, we show that lifted multiplicity codes have a better trade-off between redundancy and the number of disjoint recovering sets for every codeword or information symbol than previously known constructions, thereby providing the best known PIR  codes for some parameter regimes. Additionally, we present a new local self-correction algorithm for lifted multiplicity codes.
\vspace{-5pt}
\end{abstract}

\section{Introduction}
The concepts of \emph{locality} and \emph{availability} of codes have been subject to intensive studies.  Informally, the locality of a code refers to the number of codeword symbols that need to be accessed in order to recover a single codeword or information symbol and availability is the number of such (disjoint) recovery sets. These properties are of interest in a variety of applications, such as load balancing in distributed data storage, cryptography, and low-complexity error correction/detection. Several different notions related to these parameters have been considered in the literature, including, but not limited to, locally recoverable codes (LRCs) \cite{huang2013pyramid,GHSY12}, locally decodable/correctable codes (LDCs/LCCs)~\cite{katz2000efficiency,yekhanin2012locally}, relaxed LCCs~\cite{gur2018relaxed} and LDCs~\cite{ben2006robust}, batch codes~\cite{ishai2004batch}, PIR codes~\cite{fazeli2015pir}, and codes with the disjoint repair group property (DRGP)~\cite{li2019lifted}.

Reed-Muller (RM) codes are a popular class of codes that can provide strong locality and availability properties, as already exploited in the early majority-logic decoding algorithms \cite{reed1954class}.
These codes are defined as the evaluation of multi-variate polynomials up to a specific degree in all points of a multidimensional space. Their restriction to the evaluation points that fall on one line in this evaluation space can readily be seen to be equivalent to the evaluation of a univariate polynomial in the variable over the one-dimensional space spanned by this line. If the degree of this univariate polynomial is low, these positions form a codeword of a (non-trivial) Reed-Solomon (RS) code, another well-studied class of evaluation codes. This principle can be exploited to show locality and availability properties of the RM code, which have been subject to extensive study (see, e.g., \cite{arora2003improved,alon2005testing,rubinfeld1996robust}).  However, the obvious drawback of RM codes with nice local recovery properties is their rather low rate of $R\leq 1/2$.

To overcome this issue of low rate, the concept of \emph{lifted RS codes} was introduced in~\cite{guo2013new}. Instead of evaluating only multi-variate polynomials of a limited degree, as in RM codes, these codes consist of the evaluation of all polynomials that are equivalent to the evaluation of a low-degree univariate polynomial when restricted to a line. Using this concept of lifting, which first appeared in \cite{ben2011symmetric} in the context of LDPC codes, \cite{guo2013new} presents constructions of codes from multi-variate polynomials along with good bounds on the redundancy for the bi-variate case. These codes are of considerably higher rate than RM codes, while, broadly speaking, preserving the locality properties of the RM code. The main highlight of these codes is the construction of high-rate high-error LCCs. As a conceptual result, it was shown in~\cite{guo2013new} that any polynomial producing a codeword of the lifted RS code can be decomposed to a linear combination of \textit{good} monomials whose restriction to lines are low-degree. Thus, the code rate is equal to the \textit{fraction} of good monomials.
We remark that the distance properties of these codes follow from the fact that each symbol has many disjoint recovering sets and, thus, the relative distance of lifted RS codes is similar to the one of RM codes.

\emph{Multiplicity codes} \cite{kopparty2014high} are another recently introduced class of codes with good locality properties based on RM codes. Here, each codeword symbol not only consists of the evaluation of a degree-restricted multi-variate polynomial, but it also contains the evaluation of all the derivatives of this polynomial up to some order. Similar to the concept of lifting, this generalization provides codes with significantly better rate than RM codes, while providing good locality properties. In particular, it was proved~\cite{kopparty2014high} that multiplicity codes represent a family of high-rate LCCs that have very efficient local decoding algorithms. The analysis of the rate of multiplicity codes is rather straightforward, whereas distance properties are implied by a bound on the number of points that a low-degree polynomial can vanish on with high multiplicity.

As both lifted RS codes and multiplicity codes are based on generalizations of RM codes, it is a natural question whether these techniques can be combined to further improve the parameters of the respective codes. Some progress in the study of these \emph{lifted multiplicity codes} has recently been made in \cite{wu2015revisiting,li2019lifted}. In \cite{wu2015revisiting}, the authors show asymptotic results for any number of variables. Paper \cite{li2019lifted} is devoted to improving the existing bounds on the required redundancy in the bi-variate case.

\subsection{Our contribution}
In this work we continue the study of lifted RS codes and lifted multiplicity codes by generalizing the results on the bi-variate case of \cite{li2019lifted,guo2013new} to an arbitrary number of variables. Since lifted RS codes represent a specific class of lifted multiplicity codes, when derivatives are not taken into account, we focus on the description of lifted multiplicity codes in the following. Essentially, we investigate the same class of codes as defined in~\cite{li2019lifted,wu2015revisiting}. Informally, the $[m,s,d,q]$ lifted multiplicity code consists of the evaluation (together with the derivatives up to the $s$th order) of polynomials from $\F_q[X_1,\ldots,X_m]$ whose restriction to a line agrees with some polynomial of degree less than $d$ on its first $s-1$ derivatives. Note that the condition $d<qs$ guarantees~\cite{wu2015revisiting,li2019lifted} that the all-zero codeword is produced only by the zero polynomial and, therefore, we fix $d=qs-r$ for some integer $r$.

Following a standard approach, we consider a subcode of a lifted multiplicity code which is formed by the linear span of \textit{good} monomials whose restriction to a line is equivalent to a low-degree polynomial. To count bad monomials, we first make use of the result for lifted RS codes $(s=1)$ derived in Section~\ref{ss::lifted RS codes} and then extend it for larger $s$. Roughly speaking, we prove that there exists a one-to-$\binom{s+m-1}{m-1}$ correspondence between bad monomials for lifted RS codes and groups of bad monomials for lifted multiplicity codes.  This enables us to find the exact asymptotic order of the number of bad monomials when $q$ is large (for more details, see Section~\ref{ss::rate and distance of lifted multiplicity codes}). Unfortunately, unlike lifted RS codes, there is no nice structural result saying that a good polynomial of a lifted multiplicity code can be decomposed into a linear combination of good monomials (for a counterexample see Appendix~\ref{ss::equivalenceLiftedRS}). However, the fraction of good monomials serves as a lower bound on the rate of a lifted multiplicity code. Compared to prior works, our estimate for lifted RS codes is consistent with~\cite{guo2013new} for $m=2$, with~\cite{polyanskii2019lifted} for $m=3$, and better than the result of~\cite{guo2013new} for any $m>2$. As for lifted multiplicity codes, our estimate is consistent with~\cite{li2019lifted} for $m=2$ and better than the result of~\cite{wu2015revisiting} for any $m\ge 2$.

Let $\binom{m}{\ge b}$ denote the number of ways to choose an (unordered) set of at least $b$ elements from a fixed set of size $m$. Our main contribution is summarized in the following statement.
\begin{theorem*}[Parameters of lifted multiplicity code]\ \\
\textbf{Code rate:} For powers of two $q$ and $s<q$ and a positive integer $r<q$, the rate of the $[m,s,qs-r,q]$ lifted multiplicity code is
\begin{align*}
1 - O_m\left(s^{-1}(q/r)^{\log \lambda_m - m}\right)\quad \text{as } q\to\infty,
\end{align*}
where $\lambda_m$ is the largest eigenvalue of the matrix $A_m$ defined as
\begingroup
\setlength{\arraycolsep}{2pt}
\begin{align*}
\left(\begin{array}{cccccc}
\binom{m}{\ge 1} & \binom{m}{ 0} & 0 & 0 & \dots & 0 \\
\binom{m}{\ge 3} & \binom{m}{2} &\binom{m}{1} & \binom{m}{0} & \dots & 0 \\
\vdots & \vdots & \vdots & \vdots & \ddots & \vdots \\
\binom{m}{\ge 2j+1} & \binom{m}{ 2j} & \binom{m}{ 2j-1} & \binom{m}{ 2j-2} & \dots & \binom{m}{ 2j-m+2}
\\
\vdots & \vdots & \vdots & \vdots & \ddots & \vdots \\
\binom{m}{\ge 2m-1} & \binom{m}{2m-2} & \binom{m}{ 2m-3} & \binom{m}{ 2m-4} & \dots & \binom{m}{ m}
\end{array}\right).
\end{align*}
\endgroup
\textbf{Distance:} For $r,s<q$, the relative distance $\Delta$ of the $[m,s,qs-r,q]$ lifted multiplicity code is
\begin{align*}
\Delta\geq \Delta_{min}\coloneqq\left \lceil\frac{r-s+1}{s}\right\rceil\frac{q-s}{q^2}.
\end{align*}
For $s=o(r)$, $\Delta_{min}= \frac{r}{qs}(1+o(1))$.\\
\textbf{Availability:} Each symbol of a codeword of the $[m,s,qs-s,q]$ lifted multiplicity code can be reconstructed in $\lfloor q/s \rfloor^{m-1}$ different ways, each of which involves a disjoint set of coordinates of the codeword with cardinality $s^{m-1}(q-1)$.\\
\textbf{Local self-correction:} For $s^{m-2}=o(\log q)$ and $r<q$, let $\y$ be a noisy version of a codeword $\c$ of the $[m,s,qs-r,q]$ lifted multiplicity code such that the relative distance $\Delta(\y,\c)< \alpha\Delta_{min}$ with $0<\alpha<1/4$. Then for any $i\in[q^m]$, there exists a randomized algorithm $\A$ that makes at most $(q-1)s^{m-1}$ queries to $\y$ and reconstructs $c_i$ correctly with probability at least $1-2\alpha + o(1)$.

\end{theorem*}
We have several additional remarks and comments illustrating the contribution of our paper.
\begin{itemize}
 \item The advantage of moving from lifted RS codes to lifted multiplicity codes is that the redundancy improves by a factor of $s$ (the order of derivatives), while the number of repair groups gets worse by a factor of $s^{m-1}$ and the logarithm of the alphabet size increases by a factor of ${s + m - 1 \choose m}$. This means that lifted multiplicity codes cover more parameters of codes with good locality properties than lifted RS codes. For a relevant comparison, see the remarks after Lemmas~\ref{lem::known1}-\ref{lem::known2}.

 \item Let us demonstrate the improvement in the  rate of the $[m,s,qs-r,q]$ lifted multiplicity codes compared to the rate of the multiplicity code of order-$s$ evaluations of degree $qs-r$  polynomials in $m$ variables over $\F_q$~\cite[Lemma 7]{kopparty2014high}. Both types of codes have the same estimate on the relative distance $\Delta \ge \frac{r}{qs}(1+o(1))$. However, the rate of the multiplicity code is
\begin{align*}
\frac{\binom{qs-r+m}{m}}{\binom{s+m-1}{m}q^m} < \left(\frac{qs-r+m}{(s+1/3) q}\right)^m \le 1 - \Omega_m\left(s^{-1}\right),
\end{align*}
which is smaller than the rate of lifted multiplicity codes as $\log \lambda_m< m$. Here, we point out that for large $m$, we are able to find the technical parameter $\lambda_m$ numerically only. We depict some values of $\lambda_m$ in Table~\ref{tab::eigenvalues}. This parameter stands for the exponential growth of the number of bad monomials. The inequality $\log \lambda_m< m$ follows from~\cite{guo2013new} implicitly, as the true exponent $\log \lambda_m$ was estimated by $m-p_m<m$, where $p_m:=-\log \left(1-2^{-m\lceil \log m\rceil}\right)/\lceil \log m \rceil$. On the other hand, it is possible to estimate $\log \lambda_m$ from the other side as follows
\begin{align}\label{eq::bounds on top eigenvalue}
p_m \le m - \log \lambda_m \le -\log(1-2^{-m})
\end{align}
and, thus, $m-\log \lambda_m >0$ vanishes as $m\to\infty$.
\begin{table}
  \centering
  \caption{The largest eigenvalue $\lambda_m$ of $A_m$, the resulting convergence rate $m-\log(\lambda_m)$ derived in Section~\ref{ss::lifted RS codes}, and the convergence rate $p_m$ of \cite{guo2013new} for different values of $m$. }
  \begin{tabular}{CCCC}
    m & \lambda_m & m-\log(\lambda_m) & p_m\\ \hline
    2 & 3.0000 & 4.1504 \times 10^{-1}& 4.1504 \times 10^{-1}\\
    3 & 7.2361 & 1.4479 \times 10^{-1}& 1.1360 \times 10^{-2}\\
    4 & 15.5436 & 4.1747 \times 10^{-2} & 2.8233 \times 10^{-3}\\
    5 & 31.7877 & 9.6043 \times 10^{-3} & 4.6986 \times 10^{-4}\\
    6 & 63.9217 & 1.7653 \times 10^{-3}& 1.1742 \times 10^{-4}\\
    7 & 127.9763 & 2.6714\times 10^{-4}& 2.9353 \times 10^{-5}\\
    8 & 255.9939 & 3.4467 \times 10^{-5}& 2.8664 \times 10^{-8}\\
    9 & 511.9986 & 3.8959 \times 10^{-6} & 2.6872 \times 10^{-9}\\
    10 & 1023.9997 & 3.9323 \times 10^{-7}& 3.3590 \times 10^{-10}
  \end{tabular}
  \label{tab::eigenvalues}
\end{table}

\item Observe that if a good polynomial and its derivatives do not vanish on a point, then it is still possible that the restrictions of the polynomial to some lines containing this point are equivalent to the zero polynomial. This fact was overlooked in~\cite{wu2015revisiting} when proving the distance property of lifted multiplicity codes. However, we can always say that the restriction of the polynomial to at least $(q-s)q^{m-2}$ lines crossing this point is equivalent to a non-zero univariate polynomial of degree less than $qs-r$ and, thus, the minimum distance of the code is at least $1+\lceil r/s - 1\rceil (q-s)q^{m-2}$ (for more details, see Section~\ref{ss::rate and distance of lifted multiplicity codes}).

 \item Note that the self-correction algorithm for multiplicity codes from~\cite{kopparty2014high} works well for lifted multiplicity codes. However, for small enough $s$, we present a slightly different local self-correction algorithm which requires $s\, 5^m$ times less locality. Here we combine two ideas: 1) for recovering of the evaluation of a polynomial and its derivatives up to order $s$ at a point, it is sufficient to know directional derivatives for $s^{m-1}$ lines containing the point whose directional vectors $(1,v_2,\ldots,v_m)$ form a subcube $1\times Q_2\times\dots\times Q_m$ with $Q_i\subset \F_q$, $|Q_i|=p$; 2) every $(m-1)$-uniform hypergraph with $q$ vertices in each part with at least $\epsilon q^{m-1}$ hyperedges contains a copy of $(m-1)$-uniform clique with $s$ vertices in each part (for more details, see Section~\ref{ss::locally correctable codes}).
\end{itemize}

\subsection{Outline}
The remainder of the paper is organized as follows. In Section~\ref{ss::preliminaries}, we give rigorous definitions of lifted RS codes and lifted multiplicity codes along with some auxiliary notation. The rate of lifted RS codes can be determined by computing the fraction of so-called \emph{good monomials}, for which we will derive tight asymptotic formulas in Section~\ref{ss::lifted RS codes}.
Using the latter result, in Section~\ref{ss::lifted mult codes}, we derive bounds on the rate and distance of lifted multiplicity codes. In Sections~\ref{ss:: PIR codes},~\ref{ss:: batch codes}, and~\ref{ss:: LCCs}, we apply the results of Sections~\ref{ss::lifted RS codes} and~\ref{ss::lifted mult codes} results to PIR codes, batch codes, and LCCs, respectively. %
Finally, we conclude with Section~\ref{ss::conclusion}.

\section{Preliminaries}\label{ss::preliminaries}

\subsection{Notation}\label{ss::notation}

We start by introducing some notation that is used throughout the paper. For some functions $f(x)$ and $g(x)$, we write $f(x)=O(g(x))$ and $f(x)=\Omega(g(x))$ as $x\to\infty$ if there exists some real $x_0$ and $C$ such that $|f(x)|\le C|g(x)|$ and $|f(x)|\ge C|g(x)|$ for $x\ge x_0$, respectively. If both equalities $f(x)=O(g(x))$ and $f(x)=\Omega(g(x))$ hold, we use the notation $f(x)=\Theta(g(x))$. Also, we write $f(x)=o(g(x))$ as $x\to\infty$ if for every positive $\epsilon$ there exists some real $x_0$ such that $|f(x)|\le \epsilon|g(x)|$  for $x\ge x_0$. In these notations, we use a subscript, such as $O_m(f(x))$, if the parameter $m$ is to be regarded as fixed.

Let $[n]$ be the set of integers from $1$ to $n$. We use uppercase letters such as $T$ and $X$ to denote variables. A vector is denoted by bold letters, e.g., $\d$ is a vector over a field or a ring and $\X$ is a vector of variables.  Let $q=2^{\ell}$ and $\F_q$ be a field of size $q$. We write $\log x$ to denote the logarithm of $x$ in base two. By $\Z_{\ge}$ and $\Z_{n}$  denote the set of non-negative integers and the set of integers between $0$ and $n-1$, respectively. In what follows, we fix $m$ to be a positive integer representing the number of variables. For $\d = (d_1,\dots, d_m)\in \Z_{q}^m$ and $\X=(X_1,\dots,X_m)$, let $\X^\d$ denote the monomial $\prod\limits{_{i=1}^m} X_i^{d_i}$ from $\F_q[\X]$. Let $\deg(\d)$ be the sum of components of $\d\in\Z_{\ge}^n$ and $|\d|$ be the number of non-zero components of $\d$. Additionally, we define $\deg_q(\d)\eqdef \sum_{i=1}^m \lfloor d_i / q \rfloor$. For a vector $\i \in \Z_{\ge}^m$, let $[\X^\i]f(\X)$ denote the coefficient of $\X^\i$ in the polynomial $f(\X)$. For $f(\X)\in\F_q[\X]$, we define $\deg(f)$ to be the maximal $\deg(\i)$ for $\i$ such that $[\X^\i]f(\X)$ is non-zero.

Let us define a partial order relation on $\Z_{q}$. For two integers $a=\sum_{i=0}^{\ell-1} a^{(i)} 2^i$ and $b=\sum_{i=0}^{\ell-1} b^{(i)} 2^i$ with $a^{(i)},b^{(i)}\in\{0,1\}$ we write $a\le_2 b$
if $a^{(i)}\le b^{(i)}$ for all $i\in\{0,\dots, \ell-1\}$. We denote $a=(a^{(\ell-1)},...,a^{(0)})_2$.
For vectors $\d, \d'\in \Z_{q}^m$, we write $\d\le_2 \d'$ if $d_i\le_2 d_i'$ for all $i\in[m]$. %

Define the function $\Modsp{q}{s}: \Z_{\ge}\to \Z_{qs}$ that takes a non-negative integer $a$ and maps it to $a \Modsp{q}{s}$ by the rule:  if $a\in\Z_s$, then $a \Modsp{q}{s} = a$; if $a\ge s$ and $a=b \Mod{qs-s}$ with $b\in \Z_{qs}\setminus \Z_s$, then  $a \Modsp{q}{s} = b$.
If $s=1$, we drop the index and write $\!\! \Mods{q}$ instead of $\!\! \Modsp{q}{1}$.
It can be readily seen that if $a\,(\text{mod}^* q)=b$, then $T^{a} = T^{b} \Mod{T^q-T}$ in $\F_q[T]$. A similar equivalence for $\Modsp{q}{s}$ will be defined in Section~\ref{ss::def lifted multiplicity codes}.

For a function $f: \F_q^m \to \F_q$ and a set $S\subset \F^m_q$, let $f|_S$ denote the restriction of $f$ to the domain $S$. Abbreviate the set of all lines in $\F_q^m$ by
\begin{align*}
\L_m\eqdef\left\{(\w+ \v T)|_{T\in\F_q} \text{ for }\w,\v\in\F_q^m \right\}.
\end{align*}

We note that a multivariate polynomial restricted to a line is a univariate polynomial and the degree of the latter does not depend on the parameterization of the line, i.e., the degree of the univariate polynomial obtained by restricting to a line $L = (\w + \gamma_1 \v +\gamma_2\v T)|_{T \in \F_q}$ with $\gamma_1 \in \F_q$ and $\gamma_2 \in \F_q^*$ is independent of the choice of $\gamma_1$ and $\gamma_2$.
Denote the set of univariate polynomials of degree less than $d$ by
\begin{align*}
\cF_{q}(d)\eqdef\{f(T)\in\F_q[T]:\,\,\deg(f)< d\}.
\end{align*}

\subsection{Lifted Reed-Solomon codes}

Let us recall the definition of lifted Reed-Solomon codes introduced in~\cite{guo2013new}.%
\begin{defn}[Lifted Reed-Solomon code,~%
\cite{guo2013new}]\label{def::lifted RS code}
For integers $m\ge 1$ and $d< q$, the \textit{$m$-dimensional lift of a Reed-Solomon code} (or the \textit{$[m,d,q]$ lifted RS code})  is the code
\begin{align*}
\left\{(f(\a))|_{\a\in\F_q^m}:
\begin{aligned}
f(\X)\in\F_q[\X]\text{ s.t. }
\forall L\in \L_m:f|_L\in\cF_{q}(d)
\end{aligned}\right\}.
\end{align*}
\end{defn}
\begin{remark}
Note that the one-dimensional lift of a Reed-Solomon code represents the ordinary Reed-Solomon code of length $q$ and dimension $d$. Also, we observe that the $[m,d,q]$ lifted RS code includes all codewords of the $m$-variate RM code of order $d-1$ over $\F_q$.
\end{remark}
In Appendix~\ref{ss::improvement RM vs LRS}, we provide a simple example which demonstrates that there exist polynomials contained in the lifted RS code, that are not of low-degree, i.e., not contained in the respective RM code.

\begin{defn}[$d^*$-bad and good monomials] \label{def::bad * monomial}
Given a positive integer $d< q$, we say that a monomial $\X^\d$ with $\d\in\Z_q^m$ is \textit{$d^*$-bad} over $\F_q[\X]$ if there exists at least one $\i\in\Z_{q}^m$ such that $\i\le_2\d$ and $\deg(\i) \Mods{q} \in \{d,d+1,\dots,q-1\}$. A monomial is said to be \textit{$d^*$-good} if it is not $d^*$-bad.
\end{defn}
A characterization of lifting was established in~\cite{guo2013new}. We make use of this result for lifted Reed-Solomon codes.
\begin{lemma}[Follows from~{\cite[Section 2]{guo2013new}}]\label{lem::lifted Reed-Solomon codes}
The $[m,d,q]$ lifted RS code is equivalently defined as the evaluation of polynomials from the linear span of $d^*$-good monomials over $\F_q[\X]$.
\end{lemma}
Lemma~\ref{lem::lifted Reed-Solomon codes} suggests a way to compute the dimension of the $[m,q,d]$ lifted RS code, namely one needs to estimate the size of the set of $d^*$-good $m$-variate monomials over $\F_q[\X]$. We carry out a careful analysis of the latter in Section~\ref{ss::lifted RS codes}.

\subsection{Lifted multiplicity codes}\label{ss::def lifted multiplicity codes}
\begin{defn}
  For $f(\X)\in \F_q[\X]$ and a vector $\i\in\Z_{\ge}^m$, the $\i$th \textit{(Hasse) derivative} of $f$, denoted by $f^{(\i)}(\X)$, is the coefficient $[\Y^\i]g(\X,\Y)$, where the
  polynomial $g(\X,\Y)\coloneqq f(\X+\Y)\in \F_q[\X,\Y]$. Therefore, we have
  \begin{align*}
    g(\X,\Y) = \sum_{\i\in\Z_{\ge}^m}
    f^{(\i)}(\X)\Y^\i.
  \end{align*}
\end{defn}
For an $\x\in \F_q^{m}$, an integer $s\ge1$, and a polynomial $f(\X) \in \F_q[\X]$, we write $f^{(<s)}(\x)\in \F_q^{\binom{s+m-1}{m}}$ to denote the vector containing $f^{(\i)}(\x)$ for all $\i\in\Z_{\ge}^m$ so that $\deg(\i) < s$. In what follows, we assume that $s$ is a power of two.

We recall two well-known properties of the Hasse derivative which will imply the linearity of lifted multiplicity codes over $\F_q$.
\begin{proposition}
  Let $f(\X),\ g(\X)\in \F_q[\X]$, $\lambda\in \F_q$ and let $\i\in \Z_{\ge}^{m}$. Then we have
  \begin{enumerate}
  \item  $f^{(\i)}(\X) + g^{(\i)}(\X) = (f + g)^{(\i)}(\X).$
  \item $(\lambda f)^{(\i)}(\X)=\lambda f^{(\i)}(\X).$
  \end{enumerate}
\end{proposition}
\begin{defn}\label{def::equivalenceOrderP}
  We say that two univariate polynomials $f(X),g(X)\in \F_q[X]$ are equivalent up to order $s$ if $f^{(<s)}(x) = g^{(<s)}(x)$ for all $x\in\F_q$. To indicate such an equivalence, we write $f(X) \equiv_s g(X)$.
\end{defn}
The following statement shows the smallest possible degree of an equivalent polynomial.
\begin{proposition}[Lemma 12 in~\cite{li2019lifted}]\label{pr::reducing the power}
  Let $q$ be a power of two. For every univariate polynomial $f(X)$, there exists a unique degree-at-most $sq-1$ polynomial $g(X)$ such that $f(X) \equiv_s g(X)$. Moreover, if $s$ is a power of two, then $f(X) = g(X) \pmod{X^{qs}+X^s}$ and for all $i$ such that $\deg(f) - qs +s < i < qs$, we have $[X^i]f(X)=[X^i]g(X)$.
\end{proposition}
If $s$ is a power of two and $a\,\Modsp{q}{s}=b$, then $T^{a} \equiv_s T^{b}$.
Now we give a well-known result about the multiplicities of a multi-variate polynomial.
\begin{lemma}[Follows from~\cite{dvir2013extensions}]\label{lem::multiplicity}
 Let $f(\X)$ be a non-zero polynomial of degree at most $d$. Then the number of points $\x\in\F_q^m$ such that $f^{(\i)}(\x)=0$ for all $\i\in\Z_{\ge}^m$ with $\deg(\i)<s$ is at most $\lfloor d q^{m-1} /s\rfloor$.
\end{lemma}

\begin{defn}[Lifted multiplicity code~\cite{li2019lifted}]\label{def::lifted mult code}
  For integers $m\ge 1$ and $d<qs$, the  $[m,s,d,q]$ \textit{lifted multiplicity} code over $\F_q^{\binom{s+m-1}{m}}$ of length $q^m$ is defined as
  \begin{equation*}
    \left\{
      \left.\left(f^{(<s)}(\a)\right)\right|_{\a \in\F_q^m} \ : \
      \begin{aligned}
        &f(\mathbf{X}) \in \F_q[\mathbf{X}]\ \text{such that} \\
        & f|_L \equiv_s g(T) \ \forall \ L=L(T) \in \mathcal{L}_m\\
        &\text{for some} \ g \in \mathcal{F}_{q}(d)
      \end{aligned}
    \right\} .
  \end{equation*}
\end{defn}
\begin{remark}
Multiplicity codes, as defined in \cite{kopparty2014high}, consist of the evaluations of multi-variate polynomials of degree less than $d$. These polynomials trivially fulfill the condition that their restriction to every line $L \in \mathcal{L}_m$ is a polynomial of degree less than $d$. It follows that the $[m,s,d,q]$ multiplicity code is a subcode of the $[m,s,d,q]$ lifted multiplicity code. Thereby, the dimension of a lifted multiplicity code is lower bounded by the dimension of the corresponding multiplicity code. However, for many parameters, lifting increases the rate of the multiplicity code, as we formally show in Section~\ref{ss::lifted mult codes}. To provide some further intuition, we also give an example for this improvement in Appendix~\ref{ss::improvement LMC vs MC}.
\end{remark}
\begin{defn}[$(d,s)^*$-bad and good monomials] \label{def::bad (d^*,s) monomial}
  Given positive integers $s$ and $d$, we say that a monomial $\X^\d$ with $\d\in\Z_{qs}^m$ and $\deg_q(\d)\le s-1$ is \textit{$(d,s)^*$-bad} over $\F_q[\X]$ if there exists at least one $\i\in\Z_{qs}^m$ such that $\i\le_2\d$ and $\deg(\i) \Modsp{q}{s} \in \{d,d+1,\dots,qs-1\}$. A monomial $\X^\d$ with $\d\in\Z_{qs}^m$ and $\deg_q(\d)\le s-1$  is said to be \textit{$(d,s)^*$-good} if it is not $(d,s)^*$-bad.
\end{defn}
Let $\mathcal{F}_{q}(m,s,d)$ be the collection of $(d,s)^*$-good $m$-variate monomials from $\F_q[\X]$.

\begin{proposition}\label{prop::code cardinality}
For $s\le q$ and $d<qs$, the cardinality of the $[m,s,d,q]$ lifted multiplicity code is  at least $ q^{|\mathcal{F}_{q}(m,s,d)|}$.
\end{proposition}
\begin{proof}
The full proof of this technical statement is given in Appendix~\ref{ss::injective map}. There we show that different linear combinations of good monomials produce different codewords and that these codewords are contained in the $[m,s,d,q]$ lifted multiplicity code. Thus, the lower bound on the dimension of the code follows directly from the number of good monomials $|\mathcal{F}_{q}(m,s,d)|$.
\end{proof}
\begin{remark}
Observe that for $s=1$, Definition~\ref{def::lifted mult code} gives exactly the code spanned by the evaluation of good monomials, i.e., the statement of Proposition~\ref{prop::code cardinality} holds with equality. This case corresponds to lifted RS codes, for which this equivalence first appeared in~\cite{guo2013new}, as restated in Lemma~\ref{lem::lifted Reed-Solomon codes}. Therefore, we will also refer to the $[m,1,d,q]$ lifted multiplicity code as the $[m,d,q]$ lifted RS code in the following.
In Appendix~\ref{ss::equivalenceLiftedRS}, we provide some codewords of a lifted multiplicity code with $s\geq 2$, which are not included in the subcode spanned by the evaluation of good monomials, thereby showing that the statement of Proposition~\ref{prop::code cardinality} does not hold with equality in general.
\end{remark}

\section{Analysis of lifted RS codes}\label{ss::lifted RS codes}

In this section, we investigate the code dimension of lifted RS codes. For this purpose, we first introduce the concept of $(q-r)$-bad monomials (slightly different from $(q-r)^*$-bad monomials) and derive an explicit evaluation formula to count the number of such monomials when the parameter $r\le m$ is fixed and the field size $q=2^\ell$ is scaled. To emphasize that we scale $q$ independently of $r$, we do not denote the maximum degree by $d$ in the following, but instead explicitly write $q-r$. Second, we show how to use the evaluation formula to derive a bound on the number of $(q-r)^*$-bad monomials for arbitrary $r\le q$. Our estimate improves upon the result presented in~\cite[Sections 3.2, 3.4]{guo2013new} for $m\ge 3$ and is consistent with the result for $m=3$ provided in~\cite{polyanskii2019lifted}.
\subsection{Computing  the number of $(q-r)$-bad monomials}
Let us introduce a terminology useful for establishing the number of $d^*$-bad monomials. Let $r\le \min(m,q)$ be a fixed positive integer.
\begin{defn}[$(q-r)$-bad monomial]\label{def:: q - r bad}
We say that a monomial $\X^\d$ with $\d\in\Z_q^m$ is \textit{$(q-r)$-bad} over $\F_q[\X]$ if there exists at least one $\i\in\Z_q^m$ such that $\i\le_2\d$ and $\deg(\i) \pmod{q}=(q-r)$.
\end{defn}
\begin{remark}
The difference with Definition~\ref{def::bad * monomial} is, roughly speaking, in the modulo operation, namely $\!\! \pmod{q}$  is used in Definition~\ref{def:: q - r bad}, whereas $\!\! \pmod{q-1}$ is used in Definition~\ref{def::bad * monomial}.
\end{remark}
Let $S_j(\ell)$ denote the set of tuples $\d\in\Z_q^m$, $q=2^\ell$, for which there exists $\i\le_2\d$ with $\deg(\i)= (q-r) + jq = (2^{\ell}-r)+j2^{\ell}$ and $s_j(\ell)$ be the cardinality of $S_j(\ell)$. We note that $S_j(\ell)$ also depends on $r$, however, we omit this in our notion as we fix $r$ and scale only $\ell=\log q$. Also, the evaluation formula we provide does not depend on $r$. Clearly, $s_{j}(\ell)=0$ for $j\ge m$ as the maximal $\deg(\i)$ over admissible $\i$ is $m(q-1)$ which is smaller than $(q-r)+mq$. Therefore, we aim to compute $\sum_{i=0}^{m-1}s_i(\ell)$ since the number of $(q-r)$-bad monomials over $\F_q$ is bounded by this value from one side and by $s_0(\ell)$ from the other side.

\begin{example}
For $q=4$, $r=1$ and $m=2$ the set $S_0(2)$ is
\begingroup
\setlength{\arraycolsep}{0.2pt}
\begin{align*}
\begin{array}{ccccccccccccc}
    S_0(2) &=\! & \{\! &(3,\! 0),&(2,\! 1),&(3,\! 1),&(1,\! 2),&(3,\! 2),&(0,\! 3),&(1,\! 3),&(2,\! 3),&(3,\! 3)& \} \\
    &&& \downarrow \hphantom{,} & \downarrow \hphantom{,}&\downarrow \hphantom{,}&\downarrow \hphantom{,}&\downarrow \hphantom{,}&\downarrow\hphantom{,} &\downarrow \hphantom{,}&\downarrow\hphantom{,} &\downarrow\hphantom{,} & \\
     \mathbf{i} &:& &(3,\! 0)\hphantom{,} & (2,\! 1)\hphantom{,} & (3,\! 0)\hphantom{,} &(1,\! 2)\hphantom{,} & (3,\! 0)\hphantom{,} & (0,\! 3)\hphantom{,} & (1,\! 2)\hphantom{,} & (2,\! 1)\hphantom{,} & (3,\! 0)& .
\end{array} 
\end{align*}
\endgroup
It is easy to check that for any $\mathbf{d} \in S_0(2)$ and the corresponding $\mathbf{i}$ it holds that $\mathbf{i} \leq_2 \mathbf{d}$ and $\deg(\mathbf{i}) = (q-r) + jq = (4-1)+0\cdot 4=3$. The cardinality of the set is $s_0(2) = |S_0(2)| = 9$. For these parameters the only $\mathbf{d}$ with $\deg(\mathbf{d}) \geq q-r=3$ that is not $(q-r)$-bad is $\mathbf{d}=(2,2)$.

\end{example}

Before presenting our main technical result, we establish two important preliminary results. %
\begin{lemma}\label{lem::decrease weight}
    If $\d\in S_{j}(\ell)$ for a non-negative integer $j$, then $\d\in S_{l}(\ell)$ for any non-negative integer $l< j$.
\end{lemma}
\begin{proof}%
As $\d\in S_j(\ell)$, there exists some $\i$ such that $\i\le_2 \d$ and $\deg(\i)=(q-r)+jq=(2^{\ell}-r)+j 2^{\ell}$.
We shall prove that there exists $\i'$ such that $\i'\le_2 \i$ and $\deg(\i')=(2^{\ell}-r)+l 2^{\ell}$. This is sufficient for showing $\d\in S_l(\ell)$. To this end, we provide an iterative procedure that takes an arbitrary $\i\in \Z_q^m$ with $\deg(\i)\ge j2^\ell$ and outputs $\a\le_2 \i$ with $\deg(\a) = \deg(\i) - (j-l)2^{\ell}$ for $l\in[j]$. The procedure goes from the leading bits to the least significant ones and replaces some ones in the binary representations of $\i=(i_1,\ldots,i_m)$ by zeros.
\begin{enumerate}[labelwidth=!]
    \item \textbf{Step 1.} Let us initialize $\a\gets \i$ and $\Delta \gets (j-l)$ and $h \gets {\ell}$.
    \item \textbf{Step 2.} If $h=0$, output $\a$. Else, let $h \gets h-1$ and $\Delta\gets 2\Delta$. Compute $\delta=\Delta-\sum_{\xi=1}^m a_\xi^{(h)}$.
    If $\delta > 0$, let $\Delta\gets \Delta-\delta$ and $a_\xi^{(h)}\gets 0$ for all $\xi\in [m]$. Repeat Step 2.     Else, let  $m'$ satisfy $\Delta-\sum_{\xi=1}^{m'} a_{\xi}^{(h)}=0$ and let $a_{\xi}^{(h)}\gets 0$ for all $\xi \in [m']$. Output $\a$.
\end{enumerate}
According to the procedure, we output the correct $\a$ if we do the else-part in Step 2 at some point. Assume to the contrary that this does not happen. This means that we output the all-zero tuple at the end. However,  $\Delta=(j-l)2^\ell - \deg(\i)>0$ at the final step which contradicts with $\deg(\i)\ge j2^\ell$.
This completes the proof.
\end{proof}

\begin{example}
Consider the parameters $q=2^\ell=4$, $m=2$, $r=2$, $j=1$, and $l=0$. For the element $\mathbf{d} = (3,3) \in S_1(2)$ and $\mathbf{i} = (3,3) = (11,11)_2$ with $\mathbf{i}\leq_2 \mathbf{d}$, we will find the corresponding $\mathbf{a}$ with $\mathbf{a} \leq_2 \mathbf{i}$ and $\deg(\mathbf{a}) = \deg(\mathbf{i}) -(j-l)2^\ell = 2$.
\begin{enumerate}
    \item \textbf{Step 1.} Initialize $\mathbf{a} \gets (3,3)$ and $\Delta \gets j-l = 1$ and $h\gets \ell=2$.
    \item \textbf{Step 2.} Let $h\gets h-1 = 1$ and $\Delta \gets 2\Delta=2$. Compute $\delta = \Delta - \sum_{\xi =1}^m a_\xi^{(h)} = 0$. Since $\delta \not> 0$ we choose $m'=2$ to satisfy $\Delta-\sum_{\xi=1}^{m'} a_\xi^{(h)} = 0$ and set $a_{1}^{(1)}\gets 0$, $ a_{2}^{(1)} \gets 0$ to obtain $\mathbf{a} = (01,01)_2 = (1,1)$.
\end{enumerate}
As $\mathbf{a} \leq_2 \mathbf{i} \leq_2 \mathbf{d}$ and $\deg(\mathbf{a}) = q-r= 2$ it follows that $\mathbf{d} \in S_0(2)$.
\end{example}

Let us introduce some auxiliary functions. We define two maps $\Fdrop:\Z_{2^\ell}\to \Z_{2^{\ell-1}}$ and $\Flead:\Z_{2^\ell}\to \Z_{2}$ that take an integer $a=\sum_{i=0}^{\ell-1} a^{(i)} 2^i$ and output $a-2^{\ell-1}a^{(\ell-1)}$ and $a^{(\ell-1)}$, respectively (we either drop the leading bit in the binary representation of $a$ or output it).
We extend the maps $\Fdrop$ and $\Flead$ to $\Z_{2^\ell}^m$ in a straightforward manner by applying functions to each component of a vector $\a=(a_1,\ldots,a_m)\in\Z_{2^\ell}^m$, that is
\begin{align*}
&\Fdrop(\a)=(\Fdrop(a_1),\dots,\Fdrop(a_m)),\\ &\Flead(\a)=(\Flead(a_1),\dots,\Flead(a_m)).
\end{align*}
For an integer $a$, we denote $\max(a,0)$ by $(a)^+$.

\begin{lemma}\label{lem::drop the leading bit}
If $\d\in S_{j}(\ell+1)$ for a non-negative integer $j$, then $\Fdrop(\d)$ belongs to $S_{0}(\ell), S_{1}(\ell),\dots,S_{(2j+1-|\Flead(\d)|)^{+}}(\ell)$.
\end{lemma}
\begin{proof}%
By definition, if $\d\in S_{j}(\ell+1)$, then there exists some $\i\in\Z_{2^{\ell+1}}^m$ with $\i\le_2\d$ and $\deg(\i)=(2^{\ell+1}-r)+j2^{\ell+1}$. If the leading bits in $\i$ are dropped, then the sum of components of $\Fdrop(\i)$ is
\begin{align*}
\deg(\Fdrop(\i)) &= \deg(\i)-|\Flead(\i)|2^{\ell}\\
&= (2^{\ell}-r)+(2j+1-|\Flead(\i)|)|2^{\ell}.
\end{align*}
 Since we also have the property $\Fdrop(\i)\le_2 \Fdrop(\d)$, we obtain that $\Fdrop(\d)$ belongs to $S_{2j+1-|\Flead(\i)|}(\ell)$. Additionally, we note that $|\Flead(\i)|\le \min(2j+1,|\Flead(\d)|)$ as $\i\le_2\d$ and $\deg(\i)=(2^{\ell}-r)+j2^{\ell}$. From this and Lemma~\ref{lem::decrease weight}, we conclude that $r(\d)$ belongs to $S_{0}(\ell)$, $S_{1}(\ell)$, $\dots$, $S_{(2j+1-|\Flead(\d)|)^{+}}(\ell)$. This completes the proof.
\end{proof}

With these results established, we are now ready to give the key technical statement required for the estimation of the rate of lifted RS codes. Recall that $\binom{b}{\ge a}$ denotes the number of ways to choose an (unordered) subset of at least $a$ elements from a fixed set of $b$ elements. For $a<0$ or $a>b$, we assume that $\binom{b}{a}=0$.

\begin{proposition}\label{pr::recurrent formula}
The system of recurrence relations
\begin{align*}
\begin{pmatrix}
s_0(\ell+1)\\
s_1(\ell+1)\\
\vdots\\
s_j(\ell+1)\\
\vdots\\
s_{m-1}(\ell+1)
\end{pmatrix} =
 A_m
\begin{pmatrix}
s_0(\ell)\\
s_1(\ell)\\
\vdots\\
s_j(\ell)\\
\vdots\\
s_{m-1}(\ell)
\end{pmatrix}
\end{align*}
holds true, where the square $m\times m$ matrix $A_m$ is given below
\begingroup
\setlength{\arraycolsep}{2pt}
\begin{align*}
\left(\begin{array}{cccccc}
\binom{m}{\ge 1} & \binom{m}{ 0} & 0 & 0 & \dots & 0 \\
\binom{m}{\ge 3} & \binom{m}{2} &\binom{m}{1} & \binom{m}{0} & \dots & 0 \\
\vdots & \vdots & \vdots & \vdots & \ddots & \vdots \\
\binom{m}{\ge 2j+1} & \binom{m}{ 2j} & \binom{m}{ 2j-1} & \binom{m}{ 2j-2} & \dots & \binom{m}{ 2j-m+2}
 \\
\vdots & \vdots & \vdots & \vdots & \ddots & \vdots \\
\binom{m}{\ge 2m-1} & \binom{m}{2m-2} & \binom{m}{ 2m-3} & \binom{m}{ 2m-4} & \dots & \binom{m}{ m}
\end{array}\right).
\end{align*}
\endgroup
\end{proposition}

\begin{proof}
To begin, first note that we can uniquely encode $\d\in\Z_{2^{\ell+1}}^m$ by the pair $(\Flead(\d), \Fdrop(\d))$. Let us define the set $\pair(j)$ of size $s_j(\ell+1)$ as
\begin{align*}
\pair(j)=\left\{(\Flead(\d), \Fdrop(\d)):\quad \d\in S_j(\ell+1) \right\}.
\end{align*}
For $w\in\{0,\dots,m\}$, we define the set $T^{(w)}(j)$ as follows
\begin{align*}
T^{(w)}(j) =
\{(\v,\y):\,\, \v\in\Z_2^m,\y\in S_{(2j+1-w)^+}(\ell),\, |\v|=w \}.
\end{align*}
Clearly,  for different $w\in\{0,\ldots,m\}$, the sets $T^{(w)}(j)$ are pairwise disjoint, and the size of $T^{(w)}(j)$ is
\begin{equation*}\label{eq: size of Twj}
|T^{(w)}(j)|=\binom{m}{w} s_{(2j+1-w)^{+}}(\ell),
\end{equation*}
where we used the notation $s_j(\ell)=|S_j(\ell)|$. In the remaining proof, we show that the disjoint union of $T^{(w)}(j)$ coincides with $\pair(j)$, that is
\begin{align}\label{eq: disjoint union}
\pair(j)=\bigsqcup_{w\in\{0,\dots,m\}} T^{(w)}(j).
\end{align}
Note that $(2j+1-w)^{+}=0$ for $w\ge 2j+1$ and, thus, $|T^{(w)}(j)|=\binom{m}{w}s_0(\ell)$ for $w\ge 2j+1$. Combining this observation, equality~\eqref{eq: disjoint union} and the fact $s_i(\ell)=0$ for $i\ge m$ would lead to the required relation
\begin{align*}
s_j(\ell+1) &= {m \choose \ge 2j+1}s_0(\ell) +{m \choose 2j}s_1(\ell) \\
&+ {m \choose  2j-1}s_2(\ell)+\dots+{m \choose 2j-m+3}s_{m-2}(\ell)\\
&+{m \choose 2j-m+2}s_{m-1}(\ell).
\end{align*}

First, we check one direction of equation~\eqref{eq: disjoint union} -- namely, each element in $\pair(j)$ is covered by the union. Let $(\Flead(\d), \Fdrop(\d))\in \pair(j)$ for some $\d\in S_j(\ell+1)$. By denoting $w=|\Flead(\d)|$ and applying Lemma~\ref{lem::drop the leading bit}, we get that $\Fdrop(\d)\in S_{(2j+1-w)^+}(\ell)$. Therefore, $(\Flead(\d), \Fdrop(\d))\in T^{(w)}(j)$.

Second, we show that each element in $T^{(w)}(j)$ is included in $\pair(j)$. Let $(\v,\y)\in T^{(w)}(j)$. Construct $\d\in \Z_{2^{\ell+1}}^m$ to satisfy $\Flead(\d)=\v$ and $\Fdrop(\d)=\y$. By definition, we have that $|\v|=w$ and $\y\in S_{(2j+1-w)^{+}}(\ell)$. The latter means that there exists an $\i$ such that $\i\le_2 \y$ and $\deg(\i)=(2^{\ell}-r) + (2j+1-w)^{+}2^{\ell}$. Construct $\i'\in \Z_{2^{\ell+1}}^m$ such that $\Fdrop(\i')=\i\le_2 \y=\Fdrop(\d)$ and $\Flead(\i')\le_2 \v=\Flead(\d)$ and $|\Flead(\i')|=\min(2j+1,w)$. Thus, we obtain that $\i'\le_2 \d$ and $\deg(\i')=(2^{\ell+1}-r)+j2^{\ell+1}$. This completes the proof.
\end{proof}

\begin{defn}[Largest eigenvalue $\lambda_m$]\label{def::lambda}
Let $A_m$ be as in Proposition~\ref{pr::recurrent formula} and $\Lambda$ be the set of its eigenvalues. We define $\lambda_m$ to be the largest element from $\Lambda$.
\end{defn}
It is well known that the eigenvalues of a matrix are upper and lower bounded by the largest and smallest sum of its rows or columns, respectively. It follows directly from the structure of $A_m$ that $2^{m-1}\leq \lambda_m \leq 2^m$.
For the readers convenience, we provide $\lambda_m$ and $m-\log\lambda_m$ for $2\leq m \leq 10$ in Table~\ref{tab::eigenvalues}.

Note that the order of $s_j(\ell)$ is the maximum value in the matrix $A_m^\ell$, the $\ell$th power of $A_m$.  The exponential growth rate of the matrix powers $A_m^\ell$ as $\ell\to\infty$ is controlled by $\lambda_m^\ell$. Since all elements of $A_m^{m-1}$ are positive (except the $m$th row which has all zeros but the last entry), the matrix $A_m$ has only one eigenvalue of maximum modulus by Perron-Frobenius theorem for non-negative matrices (e.g., see~\cite[Theorem 8.5.2]{horn2012matrix}). Finally, we obtain the following statement.
\begin{corollary}\label{cor::number of bad monomials}
For an integer $r\le m$, the number of $(q-r)$-bad monomials is $\Theta_m(\lambda_m^{\ell})=\Theta_m(q^{\log \lambda_m})$ as $q\to\infty$.
\end{corollary}

\subsection{Computing  the number of $(q-r)^*$-bad monomials}
Now let $r\le q$ (the restriction $r\le m$ is no longer necessary, i.e., $r$ could be very large). By Definition~\ref{def::bad * monomial}, a monomial $\X^\d$ is $(q-r)^*$-bad if there exists an $\i\in\Z_{q}^m$ such that $\i\le_2 \d$ and $\deg(\i) \Mods{q} \in \{q-r,q-r+1,\dots,q-1\}$. The latter condition is equivalent to
\begin{align*}
\deg(\i)= q-r_0 + (q-1)j=(q-r_0-j) + qj
\end{align*}
for some $r_0\in[r]$ and $j\in \Z_{m}$. Let us drop the $\lceil\log(r+m)\rceil$ least significant bits in every component of $\d$ and $\i$ to obtain some $\d'$ and $\i'$ from $\Z_{q'}^m$ with $q'=2^{\ell'}$ and $\ell'=\ell-\lceil\log(r+m)\rceil$. Then we have that $\i'\le_2 \d'$ and
\begin{align*}
(q'-m)+jq'\le \deg(\i')\le \lfloor\deg(\i)/2^{\ell-\ell'}\rfloor\le (q'-1)+jq'.
\end{align*}
Therefore, by Definition~\ref{def:: q - r bad}, we have that $\X^{\d'}$ is $(q'-r')$-bad over $\F_{q'}[\X]$ for some positive integer $r'\le m$. By simple counting arguments and Corollary~\ref{cor::number of bad monomials}, the following statement is implied.
\begin{lemma}\label{lem:: number of bad * monomials}
For an integer $r<q=2^{\ell}$, the number of $(q-r)^*$-bad monomials is $\Theta_m(r^{m-\log\lambda_m} q^{\log \lambda_m})$ as $\ell\to\infty$.
\end{lemma}
\begin{proof}%
The number of $(q-r)^*$-bad monomials can be bounded by the number of $(q'-r')$-bad monomials with $r'\le m$ multiplied by the number of ways to choose $m\lceil \log(r+m)\rceil$ bits. By Corollary~\ref{cor::number of bad monomials}, it can be estimated as
\begin{align*}
m2^{m}(r+m)^m O_m\left({q'}^{\log \lambda_m}\right)=O_m\left(r^{m-\log\lambda_m} q^{\log \lambda_m}\right),
\end{align*}
where the factor $m$ comes from the number of choices for the parameter $r'\in [m]$ and $2^m(r+m)^m\ge 2^{m\lceil\log(r+m)\rceil}$ is the number of ways to choose $m\lceil \log(r+m)\rceil$ bits.

Now let us elaborate on showing that the number of $(q-r)^*$-bad monomials is $\Omega_m\left(r^{m-\log\lambda_m} q^{\log \lambda_m}\right)$. Take all $(q'-1)$-bad monomials $\X^{\d'}$ over $\F_{q'}[\X]$ with the property that there exists $\i'\le_2 \d'$ such that $\deg(\i')=q'-1$.
By Proposition~\ref{pr::recurrent formula} and Corollary~\ref{cor::number of bad monomials}, the number of such monomials can be bounded as $\Omega_m(q'^{\log \lambda_m})$. Define
\begin{align*}
\ell_0\coloneqq\lceil\log(m+r)\rceil-\lfloor \log r\rfloor.
\end{align*}
Then we concatenate every component $d'_j$ of $\d'=(d'_1,\ldots,d'_m)$  with the all-one string of length $\ell_0$
and an arbitrary binary string of length $\lfloor \log r\rfloor$. The total number of obtained tuples $\d\in \Z_q^m$ is then
\begin{align*}
2^{m\lfloor \log r \rfloor}\Omega_m\left(q'^{\log \lambda_m}\right)=\Omega_m\left(r^{m-\log\lambda_m} q^{\log \lambda_m}\right).
\end{align*}
For every resulting tuple $\d$, the monomial $\X^\d$ is also $(q-r)^*$-bad over $\F_q[\X]$. Indeed, we can construct an appropriate $\i$ based on $\i'$. To see this, we concatenate every component $i'_j$ (except $i'_1$) with the all-zero string of length $\lceil\log(r+m)\rceil$, and $i'_1$ with the all-one string of length  $\ell_0$ and the all-zero string of length $\lfloor \log r\rfloor$.

Then we have $\i\le_2\d$ and $\deg(\i)$ can be easily bounded as $q-r\le\deg(\i)\le q -1$.
This completes the proof.
\end{proof}

\begin{example}
Consider the parameters $q' =2^{\ell'} = 4$, $m=2$, $r=2$, and $q = 2^{\ell' + \lceil\log(r+m)\rceil} = 16$. As shown in the previous example, we have $\mathbf{d}' = (1,3) \in S_0(\ell')$ with $\mathbf{i}' \leq_2 \mathbf{d}'$ for $\mathbf{i}' = (1,2)$. The binary representations of $\mathbf{d}'$ and $\mathbf{i}'$ are given by
\begin{align*}
    \mathbf{d}' &= (01,11)_2  , \\
    \mathbf{i}' &= (01,10)_2 .
\end{align*}
Concatenating the all-one string of length $\ell_0 =\lceil\log(m+r)\rceil-\lfloor \log r\rfloor  =1$ followed by arbitrary strings of length $\lfloor \log r\rfloor = 1$ to the components of $\mathbf{d}'$ gives the tuples
\begin{align*}
    \mathbf{d}_1 &= (0110,1110)_2,\\
    \mathbf{d}_2 &= (0110,1111)_2, \\
    \mathbf{d}_3 &= (0111,1110)_2, \\
    \mathbf{d}_4 &= (0111,1111)_2.
\end{align*}
The $\mathbf{i}$ such that $\mathbf{i}\leq \mathbf{d}_j, \ j=1,2,3,4$, can be found by concatenating every component $i_j'$ except for $i_1'$ with $\lceil \log(r+m) \rceil= 2 $ zeros and $i_1$ with $\ell_0 = 1$ one and $\lfloor \log r\rfloor = 1$ zero, to obtain
\begin{align*}
  \mathbf{i} = (0110,1000)_2 \ .
\end{align*}
The degree of $\mathbf{i}$ is $\deg(\mathbf{i}) = 14 \geq q-r$.
\end{example}

\subsection{Code rate and distance of lifted RS codes}
\begin{theorem}%
  \label{th:: number of bad monomials}
  For a power of two $q$, the rate $R$ and the relative distance $\delta$ of the $[m,q-r,q]$ lifted RS code are
  \begin{align*}
    R = 1-\Theta_m\left((q/r)^{\log \lambda_m-m}\right), \quad \delta \ge \frac{r}{q}\quad \text{as } q\to\infty.
  \end{align*}
\end{theorem}

\begin{proof}[Proof of Theorem~\ref{th:: number of bad monomials}]
  To estimate the code rate of $[m,q-r,q]$ lifted RS codes, it suffices to compute the fraction of $(q-r)^*$-good monomials. By Lemma~\ref{lem::lifted Reed-Solomon codes} and~\ref{lem:: number of bad * monomials}, the rate is
  \begin{align*}
    1 - \Theta_m\left(r^{m-\log\lambda_m} q^{\log \lambda_m}\right)q^{-m}=1-\Theta_m\left((q/r)^{\log \lambda_m-m}\right)
  \end{align*}
  as $q\to\infty$. To estimate the relative distance of the code, we first note that the lifted RS code is linear. Suppose that $(f(\a))|_{\a \in\F_q^m}$ is a non-zero codeword. Let us say that $f(\w_0)\neq 0$. Then for any $\v\in \F_q^m\setminus \{\0\}$, the polynomial $f(\w_0+\v T)$ is equivalent to a non-zero univariate polynomial of degree at most $q-r-1$. Thus, $f(\w_0+\v t)\neq 0$ for at least $r+1$ different values $t\in\F_q$ and $f(\a)$ is non-zero for at least  $1+rq^{m-1}$ values $\a\in\F_q^m$.  This completes the proof.
\end{proof}

\section{Analysis of lifted multiplicity codes}\label{ss::lifted mult codes}

Using the results on the rate of lifted RS codes, we now move to estimating the rate and minimal distance of lifted multiplicity codes. Recall that lifted RS codes are trivial lifted multiplicity codes with $s=1$. In the following, we impose the constraint $s\ge m$ on the parameters, which helps with dropping the modulo operation in the definition of bad monomials. Then by applying the known results for lifted RS codes, we show how to find the asymptotics of the number of bad monomials when $m$ is fixed and $q$ is large. Our estimate continues the study of two-dimensional lifts initiated in~\cite{li2019lifted} and is consistent with the result for the case of $m=2$ presented there.

\subsection{Computing  the number of $(qs-r,s)^*$-bad monomials}\label{ss::computing the number of bad monomials}
In this section, we show that the number of $(qs-r,s)^*$-bad monomials can be well approximated  by ``$\binom{s+m}{m-1}$  times the number of $(q-r,1)^*$-bad monomials''.

First, we recall the known estimate for the number of $(q-r,1)^*$-bad monomials when the number of variables is fixed and the alphabet size is large, as established in Section~\ref{ss::lifted RS codes}, in the notation of lifted multiplicity codes.
\begin{corollary}%
  \label{cor:: number of bad * monomials}
  For an integer $r<q=2^{\ell}$, the number of $(q-r,1)^*$-bad monomials is $\Theta_m\left(r^{m-\log\lambda_m} q^{\log \lambda_m}\right)$ as $\ell\to\infty$, with $\lambda_m$ as in Definition~\ref{def::lambda}.
Moreover, the number of $\d\in\Z_q^m$ such that there exists an $\i\in\Z_q^m$ with $\i\le_2 \d$ and
\begin{enumerate}
	\item $\deg(\i) \pmod{q}\in\{q-r,q-r+1,\ldots,q-1\}$ is  $\Theta_m\left(r^{m-\log\lambda_m} q^{\log \lambda_m}\right)$ as $\ell\to\infty$.
	\item $\deg(\i)\in\{q-r,q-r+1,\ldots,q-1\}$ is also $\Theta_m\left(r^{m-\log\lambda_m} q^{\log \lambda_m}\right)$ as $\ell\to\infty$.
\end{enumerate}
\end{corollary}

Let $s\ge m$ be a power of two and $1\le r< q$. First, we show that for such a choice of parameters, the modulo operation in Definition~\ref{def::bad (d^*,s) monomial} can be dropped. By Proposition~\ref{pr::reducing the power}, for $f(X)\in\F_q[X]$ with
\begin{align*}
\deg(f)\le (s-1)q + m(q-1)=(m+s-1)q-m ,
\end{align*}
we have that $[X^i](f(X) \pmod{X^{qs}+X^{s}})=[X^i]f(X)$ for all $i\in \{qs-r, qs-r+1,\ldots, qs-1\}$ as
\begin{align*}
(m+s-1)q-m-qs+s=(m-1)q - m +s < qs - r.
\end{align*}
Therefore, by Definition~\ref{def::bad (d^*,s) monomial}, a monomial $\X^{\d}$ with $\d\in\Z_{qs}^m$ and $\deg_q(\d)\le s-1$ is  $(qs-r,s)^*$-bad if there exists a vector $\i$ such that $\i\le_2 \d$ and $\deg(\i)\in \{qs-r,qs-r+1,\ldots, qs-1\}$.

Let a monomial $\X^{\d}$ be $(qs-r,s)^*$-bad. Then every component of $\d$ can be represented as $d_j = \hat d_j q+ d_j'$ with $d_j'\in \Z_q$ and $\hat d_j \in \Z_s$ for all $j\in[m]$. As deduced above, there exists an $\i\in\Z_{qs}^m$ such that $\i\le_2 \d$ and $\deg(\i)\in \{qs-r,qs-r+1,\ldots, qs-1\}$. Therefore, after representing $i_j=\hat i_j q + i_j'$, we obtain that $\i'\le_2 \d'$ and $\deg(\i') \pmod{q} \in  \{q-r,q-r+1,\ldots,q-1\}$. %
Let us also check that $s-m\le \deg(\hat \d)\le s-1$. To show $\deg(\hat \d)\ge s-m$, we just note that
\begin{align*}
  \deg(\i)\le \deg(\d) &= \deg(\hat \d)q + \deg(\d')\\ &\le \deg(\hat \d)q + (q-1)m.
\end{align*}
Thus, if  $\deg(\hat \d)< s-m$, we have that $\deg(\i)\le (s-1)q - m < qs - r$ which contradicts the property  $\deg(\i)\in \{qs-r,qs-r+1,\ldots, qs-1\}$. Note that $\deg(\hat \d)=\deg_q(\d)$, therefore $\deg(\hat \d)\le s-1$.
Finally, we arrive at the following statement.
\begin{lemma}\label{lem:: number of bad (qs-r,s) monomials}
  For an integer $m<r<q=2^{\ell}$ and a power of two $s\ge m$, the number of $(qs-r, s)^*$-bad monomials is
  \begin{align*}
    \Theta_m\left(s^{m-1}r^{m-\log \lambda_m} q^{\log \lambda_m}\right)\quad \text{as }\ell\to\infty.
  \end{align*}
\end{lemma}
\begin{proof}
  As noted above, for every $(qs-r,s)^*$-bad monomial $\X^\d$, $\d$ can be uniquely decomposed to the pair $(\hat\d,\d')$, where $s-m\le \deg(\hat\d) \le s-1$ and for $\d'\in\Z_q^{m}$, there exists an $\i'\le_2 \d'$ with $\deg(\i') \pmod{q} \in  \{q-r,q-r+1,\ldots,q-1\}$. Thus, Corollary~\ref{cor:: number of bad * monomials} yields that the number of $(qs-r,s)^*$-bad monomials for $\ell\to\infty$ can be bounded by
  \begin{align*}
    \left(\sum_{j=1}^{m}\binom{s-j+m-1}{m-1}\right)& O_m\left(r^{m-\log \lambda_m} q^{\log \lambda_m}\right) \\
    &=O_m \left(s^{m-1}r^{m-\log \lambda_m} q^{\log \lambda_m}\right).
  \end{align*}
  It remains to show that this estimate is asymptotically tight.
  To see this, consider all possible $\d'\in\Z_q^m$ such that there exists $\i'\in\Z_q^m$ with $\i'\le_2 \d'$ and $\deg(\i')=q-r'\in\{q-r,q-r+1,\ldots,q-1\}$. By Corollary~\ref{cor:: number of bad * monomials} the number of such $\d'$ can be estimated as
  \begin{align*}
    \Omega_m\left(r^{m-\log \lambda_m} q^{\log \lambda_m}\right).
  \end{align*}
  Now we take a look on all possible $\hat\d\in \Z_s^m$ such that $\deg(\hat \d) = s-1$. We can estimate the number of such $\hat\d$ by
  $\binom{s+m-2}{m-1}$. For any such $
  \hat \d$, we define $\d\in\Z_{qs}^m$ to be such that $d_j=\hat d_jq + d_j'$ and note that $\X^\d$ is $(qs-r,s)^*$-bad as for $\i$ with $i_j = \hat d_jq+i_j'$, we have $\i\le_2 \d$ and
  \begin{align*}
    \deg(\i)=q\deg(\hat \d) + \deg(\i') = q(s-1) + q-r' = qs - r',
  \end{align*}
  which belongs to $\{qs-r,qs-r+1,\ldots, qs-1\}$. Therefore, the number of $(qs-r,s)^*$-bad monomials is
  \begin{align*}
    \binom{s+m-2}{m-1}\Omega_m&\left(r^{m-\log \lambda_m} q^{\log \lambda_m}\right) \\
    &= \Omega_m \left(s^{m-1}r^{m-\log \lambda_m} q^{\log \lambda_m}\right).
  \end{align*}
  This completes the proof.
\end{proof}
\subsection{Rate and distance of lifted multiplicity codes}\label{ss::rate and distance of lifted multiplicity codes}

\begin{theorem}[Rate and distance of lifted multiplicity codes]\label{th:: code rate of lifted multiplicity codes}
	For powers of two $s,q$ and integers $r$ and $m$ with $m\le s \le q$ and $r\le q$, the rate of the $[m,s,qs-r,q]$ lifted multiplicity code is
        \begin{align*}
1 - O_m\left(s^{-1}(q/r)^{\log \lambda_m - m}\right)\quad \text{as } q\to\infty.
	\end{align*}
	The relative distance $\Delta$ of the $[m,s,qs-r,q]$ lifted multiplicity code is
        \begin{align*}
\Delta\geq \Delta_{min}:=\left \lceil\frac{r-s+1}{s}\right\rceil\frac{q-s}{q^2}.
	\end{align*}
	For $s=o(r)$,  $\Delta_{min}= \frac{r}{qs}(1+o(1))$.
\end{theorem}

\begin{proof}[Proof of Theorem~\ref{th:: code rate of lifted multiplicity codes}]
	By Proposition~\ref{prop::code cardinality}, we can obtain the lower bound on the rate of the lifted multiplicity  code by computing the fraction of $(qs-r,s)^*$-good monomials. Thus, by Lemma~\ref{lem:: number of bad (qs-r,s) monomials}, the rate is
        \begin{align*}
          1-& \frac{O_m\left(s^{m-1}r^{m-\log \lambda_m} q^{\log \lambda_m}\right)}{\binom{s+m-1}{m}q^m} \\
          &\hspace{3cm}= 1 - O_m\left(s^{-1}(q/r)^{\log \lambda_m - m}\right).
	\end{align*}
	
	Now we estimate the distance of the $[m,s,qs-r,q]$ lifted multiplicity code. Consider a codeword which is the evaluation of some non-zero polynomial $f$.
	Let $\w_0\in\F_q^m$ be a coordinate such that $f^{(<s)}(\w_0)$ is not all-zero. In what follows, we prove the existence of a set $S$, $|S|\geq (q-s)q^{m-1}$, of lines containing this point such that for any $L\in S$ polynomial $f|_L$  doesn't vanish for at least $\lceil r/s\rceil$ points.  More explicitly, assume that for some $\i_0\in\Z_{\ge}^m$ with $\deg( \i_0)=i_0 < p$, $f^{(\i_0)}(\w_0)\neq 0$. Let a line $L$ be parameterized by $\w_0+T\v$ with $\v=(1,v_2,\ldots,v_m)$, $v_i\in\F_q$. Define $g_\v(T):=f|_L=f(\w_0+T\v)$. By the definition of Hasse derivatives, we have
        \begin{align*}
	g_\v(T) = \sum_{\i\in\Z_{\ge}^m}f^{(\i)}(\w_0+T\v) T^{\deg(\i)} \v^\i
	\end{align*}
	and, thus,
        \begin{align*}
		g_{\v}^{(i_0)}(0)=\sum\limits_{\i :\ \deg(\i)=i_0}f^{(\i)}(\w_0) \v^{\i}.
	\end{align*}
	Since $f^{(\i_0)}(\w_0)\neq 0$, we can think about the right-hand side of the above equality as a non-zero polynomial in $v_2,\ldots,v_m$ of degree at most $s$. This yields that there exist at most $s q^{m-2}$ different $\v=(1,v_2,\ldots, v_m)\in\F_q^m$ such that $g_\v^{(i_0)}(0)=0$. Thus, for at least $(q-s)q^{m-2}$ different lines $L$ containing the point $\w_0$, the univariate polynomial $f|_L\neq 0$. By the definition of $[m,s,qs-r,q]$ lifted multiplicity codes, for any line $L$, $f|_L$ agrees with some univariate polynomial of degree at most $qs-r-1$ on its first $s-1$ derivatives. By Lemma~\ref{lem::multiplicity}, if $g_{\v}(T)=f|_L\neq 0$, there exist at least $\lceil (r+1)/s\rceil$ points on which $f|_L$  doesn't vanish with high multiplicity, i.e., for at least $\lceil (r+1)/s\rceil$ different $t\in\F_q$, $g_\v^{(j)}(t)\neq 0$ for some $j<s$. This implies that the number of non-zero positions of the codeword produced by $f$ is at least
        \begin{align*}
	1+\left\lceil \frac{r+1}{s}-1\right\rceil (q-s)q^{m-2}.
	\end{align*}
	Since the lifted multiplicity code is $\F_q$-linear, the distance of the lifted multiplicity code can be bounded by the same value. This completes the proof.
\end{proof}

\section{PIR codes}\label{ss:: PIR codes}
In this section, we show that lifted multiplicity codes have the best known trade-off between the number of information symbols and the required redundancy for private information retrieval (PIR) codes. %

\subsection{Preliminaries and prior work}
The defining property of a $k$-PIR code is that for every message symbol, there exist $k$ mutually disjoint sets of coded symbols from which the message symbol can be uniquely recovered. PIR codes were suggested in~\cite{fazeli2015pir} to decrease storage overhead in PIR schemes preserving both privacy and communication complexity.  Formally, this family of codes is defined as follows.

\begin{defn}[PIR code, %
  {\cite{fazeli2015pir}}]\label{def::PIR code}
  Let $F:\,\Sigma^n\to\Sigma^N$ be a map that encodes a string $x_1,\dots,x_n$ to $c_1,\dots, c_N$ and $\C$ be the image of $F$.
  The code $\C$ will be called a \textit{$k$-PIR code} (or \textit{$[N,n,k]_{|\Sigma|}^{P}$ code}) over the alphabet $\Sigma$ if for every $i\in[n]$, there exist $k$ mutually disjoint sets $R_{1},\dots , R_{k}\subset[N]$ (referred to as \textit{recovering sets}) and functions $g_1,\dots, g_k$ such that for all $\c\in\C$ and for all $j\in[k]$, $g_j(\c|_{R_{j}})=x_{i}$, where $\c|_{R}$ is the projection of $\c$ onto coordinates indexed by $R$.
\end{defn}
The definition of a \textit{code with the disjoint repair group property} (DRGP)~\cite{li2019lifted} is similar to Definition~\ref{def::PIR code}, except that we should recover all codeword symbols instead of only information symbols. For $\F_q$-linear codes, any systematically encoded code with the DGRP directly gives a PIR code. In what follows, we summarize the results for PIR codes since the best known bounds for DRGP codes hold for PIR codes as well.

The main figure of merit when studying PIR codes is the value of $N$, given $n$ and $k$.  Denote by $\NP_q(n,k)$ the value of the smallest $N$ such that there exists an $[N,n,k]_q^P$ code. For the binary case, we will remove $q$ from these and subsequent notations. Since it is known that for sublinear $k$ and fixed $q$, $\lim\limits_{n\rightarrow \infty} \NP_q(n,k)/n=1$,~\cite{fazeli2015pir,guo2013new}, we evaluate these codes by their redundancy and define $\rP_q(n,k) := \NP_q(n,k)-n$.  %
In order to have a better understanding of the asymptotic behavior of the redundancy, the value of $\rP_q(n,k)$ is usually studied for either constant $k=O(1)$ or polynomial $k=\Theta(n^\epsilon)$, $\epsilon\geq 0$.

Constructions of PIR codes with fixed $k$ were first suggested in~\cite{fazeli2015pir,VRK17}. In particular, it can be seen that for $k=2$, $\rP_q(n,2) = 1$, and for any fixed $k\ge 3$,
$\rP_q(n,k) = \Theta(\sqrt{n})$~\cite{fazeli2015pir,rao2016lower,wootters2016linear}. There are several constructions of PIR codes~\cite{li2019lifted,asi2018nearly,FGW17,LC04,VRK17} and based on them, it is already possible to deduce some results on the asymptotic behavior of $\rP_q(n,k)$. For example, the constructions of \textit{one-step majority logic decodable codes} from~\cite{LC04} assure that $\rP(n,n^{\epsilon})= O(n^{0.5+\epsilon})$ for all $\epsilon\geq 0$. In~\cite{FGW17} the authors discussed partially lifted codes and their application to non-binary PIR codes. %
More results for PIR codes were achieved in~\cite{asi2018nearly} by using multiplicity codes and array codes. The construction~\cite{li2019lifted} of PIR codes is based on bi-variate lifted multiplicity codes. Constructions of PIR codes based on tri-variate lifted RS codes were investigated in~\cite{polyanskii2019lifted}. Finally, the paper~\cite{hastings2020wedge} introduced the so-called wedge-lifted codes to construct PIR codes.

\begin{lemma}\label{lem::known1}
The redundancy of non-binary PIR codes satisfies:
\begin{enumerate}
\item  \emph{Steiner systems~\cite{fazeli2015pir}:}\\ $\rP_q(n,k)= O_k(\sqrt n)$ for PIR codes with fixed $k\ge 3$.
\item  \emph{Multiplicity code~\cite{asi2018nearly}:}\\ $\rP_q(n,n^{\epsilon})= O(n^{\delta(\epsilon)})$ for $0\leq \epsilon <1 $, where $\delta(\epsilon) = 1-\frac{1}{\lfloor 2/(1-\epsilon)\rfloor}+\frac{\epsilon}{\lfloor 2/(1-\epsilon)\rfloor-1}$. \label{item:asi}
\item  \emph{Partially lifted codes~\cite{FGW17}:}\\ $\rP_q(n,n^{0.25})= O(n^{0.714})$. \label{item:FGW}
\item  \emph{Lifted mult. codes with $m=2$~\cite{li2019lifted}:}\\ $\rP_q(n,n^{\epsilon}) = O(n^{\frac{1}{2} + \epsilon (\log 3 - 1)})$ for $0\leq \epsilon < \frac{1}{2}$. \label{item:LiLMC}
\item  \emph{Lifted RS codes~\cite{guo2013new}:}\\ $\rP_q(n,n^{1-1/m})=O(n^{1+ \log \left(1-2^{-m\lceil \log m\rceil}\right)/(m\lceil \log m\rceil)})$ for an integer $m\ge 2$. \label{item:GuoQaryPIR}
\item \emph{Lifted RS codes with $m=3$~\cite{polyanskii2019lifted}:}\\  $\rP_q(n,n^{2/3})= O(n^{\log_8(5+\sqrt{5})})$. \label{item:LRS}
\end{enumerate}
\end{lemma}
\begin{remark}
From our results (c.f. Theorem~\ref{th::asymptotic non-binary disjoint repair group code}), it follows  that given $n$ and $k=n^\epsilon$, with $0<\epsilon\le 1-\frac{1}{m}$, the redundancy of non-binary $k$-PIR codes based on $m$-variate lifted multiplicity codes is $O(n^{\delta_{LM}(\epsilon,m)})$, where
  \begin{equation}\label{eq::our redundancy for PIR}
\delta_{LM}(\epsilon,m):=  \frac{m-1}{m} + \frac{1+\log \lambda_m - m}{m-1}\,\epsilon.
 \end{equation}
We remark that for $m=2$, the same $\delta_{LM}(\epsilon,m)$ was first derived in~\cite{li2019lifted} and gives the best estimate on $r(n,n^{\epsilon})$ with $0<\epsilon\le \frac{1}{2}$.
For further comparison, we provide the relevant results for the best known families of non-binary PIR codes in the same form. For $0\le\epsilon\le 1-\frac{1}{m}$, the required redundancy of $n^\epsilon$-PIR codes based on %
$m$-variate multiplicity codes (c.f.~\cite{asi2018nearly}) and $m$-variate lifted RS codes is $O(n^{\delta_{M}(\epsilon,m)})$, and $O(n^{\delta_{LRS}(\epsilon,m)})$, respectively, where  $\delta_M(\epsilon,m) := \frac{m-1}{m} + \frac{1}{m-1}\epsilon$ and $\delta_{LRS}(\epsilon,m):=\delta_{LM}(\frac{m-1}{m},m)$. Clearly, $\delta_{LM}(\epsilon,m)<\delta_M(\epsilon,m)$ for all $0<\epsilon\le 1-\frac{1}{m}$ as $\log \lambda_m-m<0$ (c.f.~\eqref{eq::bounds on top eigenvalue}).

Let us illustrate the improvement compared to Lemma~\ref{lem::known1} and consider the case of $m=3$ and $\frac{1}{2} < \epsilon \le \frac{2}{3}$. Then, we have $\delta_{LM}(\epsilon,3) = \frac{2}{3} + 0.4276 \epsilon$, $\delta_M(\epsilon,3) = \frac{2}{3} + \frac{1}{2}\epsilon$ and $\delta_{LRS}(\epsilon,3)=0.9517$. The latter is given in item \ref{item:LRS} of Lemma~\ref{lem::known1}. The proposed bound  $\delta_{LM}(\epsilon,3)$ coincides with $\delta_{LRS}(\epsilon,3)$ for $\epsilon = \frac{2}{3}$ and outperforms all known bounds for $\frac{1}{2}<\epsilon<\frac{2}{3}$.%

In Figure~\ref{fig::PIR}, we compare our results for non-binary PIR codes based on the most suitable $m$-variate lifted multiplicity codes to the known results summarized in Lemma~\ref{lem::known1}. Table~\ref{tab:comparisonNonBinaryPIR} in Appendix~\ref{app:comparison} gives the ranges in which each bound is best among all known results. It can be  verified that for any $\epsilon\in (\frac{1}{2},1)\setminus \{\frac{2}{3}\}$, our bounds based on lifted multiplicity codes improve the state-of-art results.
\end{remark}
\begin{lemma}\label{lem::known2}
The redundancy of binary PIR codes satisfies:
\begin{enumerate}
\item  \emph{Steiner system~\cite{fazeli2015pir,rao2016lower,wootters2016linear}:}\\ $\rP(n,k)\!=\! \Theta_k(\sqrt n)$ for linear PIR codes with fixed $k\!\ge\! 3$.
\item  \emph{Lifted RS codes~\cite{guo2013new}:}\\ $\rP(n,n^{1-1/m})\!=\!O(n^{1+\log \left(1-2^{-m\lceil \log m\rceil}\right)/(m\lceil \log m\rceil)}\log n)$ for an integer $m\ge 2$. \label{item:binPIRLRS}
\item  \emph{Array and one-step majority logic dec. codes~\cite{LC04, asi2018nearly}:}\\ $\rP(n,n^{\epsilon})= O(n^{0.5+\epsilon})$ for $0\leq \epsilon <1/2 $. \label{item:binPIRasi}
\item \emph{Binary image of mult. codes~\cite{asi2018nearly}:}\\ $\rP(n,n^{\epsilon})= O(n^{\delta(\epsilon)})$ for $0\leq \epsilon <1 $, where $\delta(\epsilon) =   \min\limits_{m\ge \lceil 1/(1-\epsilon)\rceil}\left\{1-\frac{m(1-\epsilon)-1}{2m(m-1)}\right\}$. \label{item:binPIRmult}
\item \emph{Binary image of lifted mult. codes with $m=2$~\cite{li2019lifted}:}\\ $\rP(n,n^{\epsilon}) = O(n^{\frac{3}{4} + \epsilon (\log 3 - \frac{3}{2})})$ for $0\leq \epsilon < \frac{1}{2}$. \label{item:binPIRLMC2}
\item \emph{Binary image of lifted RS codes with $m=3$~\cite{polyanskii2019lifted}:}\\ $\rP(n,n^{2/3})= O(n^{\log_8(5+\sqrt{5})}\log n)$. \label{item:binPIRLRS3}
\item \emph{Wedge-lifted codes~\cite{hastings2020wedge}:}\\  $\rP(n,n^{1/(2a)})=O(n^{0.5+\log(2-2^{-a})/(2a)})$ for integers $a\ge 1$. \label{item:binPIRwedge}
\end{enumerate}
\end{lemma}

\begin{remark}
The codes constructed in \cite{FGW17,holzbaur2020lifted,guo2013new,polyanskii2019lifted} are $q$-ary codes of length $N=q^m$. To obtain a binary PIR code each symbol can be converted to $\log q = \log N^{\frac{1}{m}} = \frac{1}{m} \log N = \Theta (\log n)$ symbols, hence the additional factor of $\log(n)$ in Lemma~\ref{lem::known2} compared to Lemma~\ref{lem::known1}. Clearly, the image of every recovery set of a $q$-ary symbol is also a recovery set for bit of the image of this symbol, so the number of mutually disjoint recovering sets is at least as large as in for the non-binary code. We provide the relevant results for the best known families of binary PIR codes in the same form. For $0\le\epsilon\le (m-1)/m$, the required redundancy of binary $n^\epsilon$-PIR codes based on $m$-variate lifted multiplicity codes (cf. Theorem~\ref{th::binary disjoint repair group property code}), $m$-variate multiplicity codes, and $m$-variate lifted RS codes is $O(n^{\delta'_{LM}(\epsilon)+o(1)})$, $O(n^{\delta'_{M}(\epsilon)+o(1)})$, and $O(n^{\delta'_{LRS}(\epsilon)+o(1)})$, respectively, where $\delta'_{LM}(\epsilon,m) := \frac{2m-1}{2m} + \frac{1+2\log \lambda_m - 2m}{2m-2}\epsilon$, $\delta'_M(\epsilon,m) := \frac{2m-1}{2m} + \frac{1}{2m-2}\epsilon$ and $\delta'_{LRS}(\epsilon,m):=\delta'_{LM}(\frac{m-1}{m},m)$. Therefore, computing the bounds for small $m$ and employing the inequality~\eqref{eq::bounds on top eigenvalue} for large $m$, we can range these three families of binary $n^\epsilon$-PIR codes with  $\epsilon>2/3$ as follows
$$
\min_{m\ge \lceil \frac{1}{1-\epsilon}\rceil} \! \delta'_{LM}(\epsilon,m)\!<\!\! \min_{m\ge \lceil \frac{1}{1-\epsilon}\rceil}\! \delta'_{M}(\epsilon,m) \!<\!\! \min_{m\ge \lceil \frac{1}{1-\epsilon}\rceil}\!\delta'_{LRS}(\epsilon,m).
$$
We remark that for $m=2$ the same $\delta_{LM}(\epsilon,m)$ was first derived in~\cite{li2019lifted}.

The binary image of lifted multiplicity codes requires the minimal redundancy among the best binary PIR codes, as given in Lemma~\ref{lem::known2}, in the range %
 $\epsilon \in (0.273,1)$. Our bounds provide a strict improvement for $\epsilon \in (\frac{1}{2},1) \setminus \{2/3\}$. A more detailed comparison to the known constructions is given in Table~\ref{tab:comparisonBinaryPIR} in Appendix~\ref{app:comparison}.
\end{remark}
 On the other hand, there is no lower bound on the redundancy of PIR codes other than that for $k\ge 3$ the redundancy of linear PIR codes of dimension $n$ is $\Omega(\sqrt{n})$~\cite{rao2016lower,wootters2016linear}.

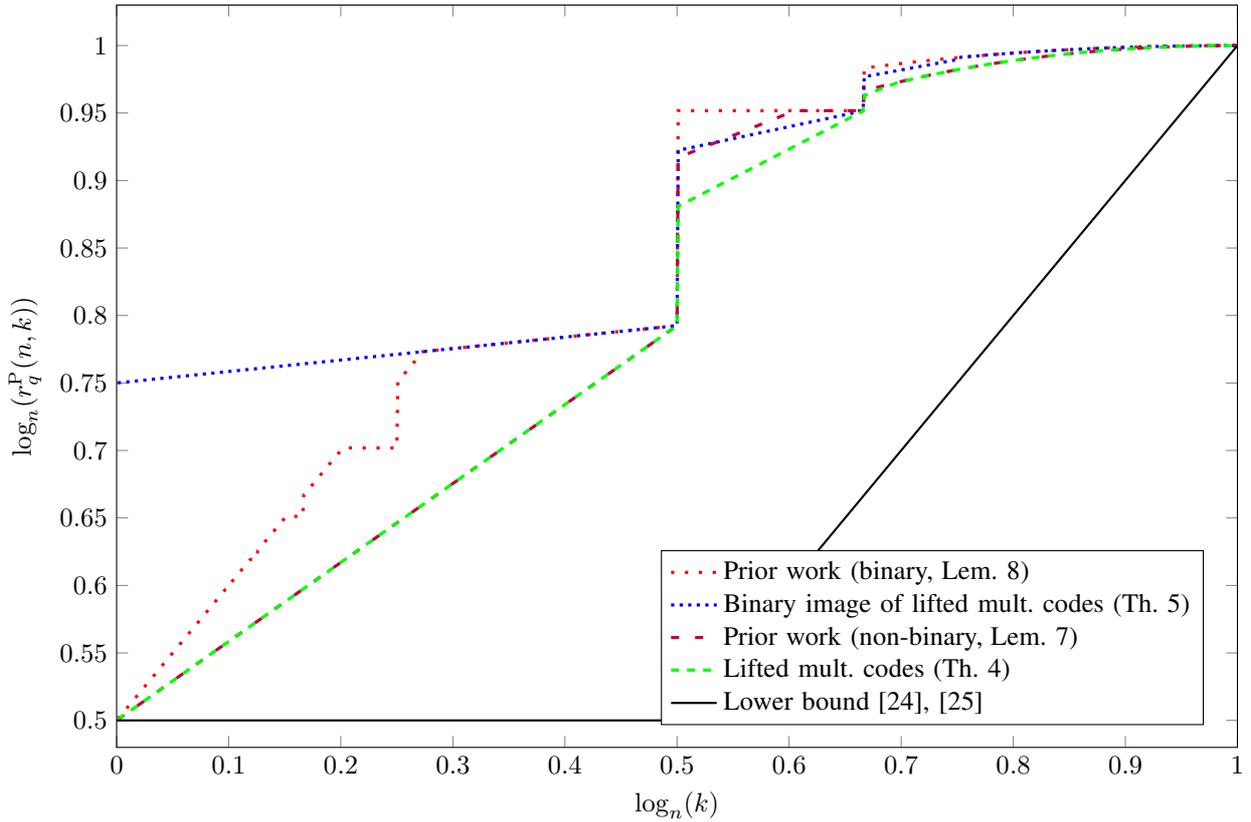
\begin{figure*}[t]
  \centering

\begin{tikzpicture}[thick,scale=1]
	\pgfplotsset{compat = 1.3}
	\begin{axis}[
		legend cell align={left},
		width = 0.9\textwidth,
		height = 0.63\textwidth,
		xlabel = {$\log_n(k)$},
		ylabel = {$\log_n(\rP_q(n,k))$},
		xmin = 0,
		xmax = 1,
		ymin = 0.48,
		ymax = 1.03,
		legend pos = south east]
		
		\addplot [ very thick, color=red,loosely dotted, mark=none] table[x=eps,y=bin] {upperBoundPIRFile.txt};
		\addlegendentry{Prior work (binary, Lem.~\ref{lem::known2})}%

		\addplot[color= blue, mark=none,dotted,very thick] table[x=eps,y=binNew] {upperBoundPIRFile.txt};
		\addlegendentry{Binary image of lifted mult. codes (Th.~\ref{th::binary disjoint repair group property code})}

		\addplot[color=purple, mark=none,loosely dashed,very thick] table[x=eps,y=qary] {upperBoundPIRFile.txt};
		\addlegendentry{Prior work (non-binary, Lem.~\ref{lem::known1})}%

		\addplot[color=green, mark=none,dashed, very thick] table[x=eps,y=qaryNew] {upperBoundPIRFile.txt};
		\addlegendentry{Lifted mult. codes (Th.~\ref{th::asymptotic non-binary disjoint repair group code})}
		
		\addplot[color=black,mark =none,thick] coordinates { (0,0.5) (0.5,0.5) (1,1) };
		\addlegendentry{Lower bound~\cite{wootters2016linear,rao2016lower}}
	\end{axis}
\end{tikzpicture}
  \caption{Comparison of parameters of binary and non-binary PIR codes based on lifted multiplicity codes to the upper and lower bounds on the minimal redundancy of \cite{asi2018nearly,li2019lifted,wootters2016linear,rao2016lower,FGW17,polyanskii2019lifted}. For $\log_n(k)\leq 0.5$ the results of Theorem~\ref{th::asymptotic non-binary disjoint repair group code} and Theorem~\ref{th::binary disjoint repair group property code} recover the results from \cite{li2019lifted}.}
  \label{fig::PIR}
\end{figure*}

\subsection{PIR and DRGP codes from lifted multiplicity codes}\label{ss:: PIR codes from LMC}
In this section, we apply our results on lifted multiplicity codes established in Section~\ref{ss::lifted mult codes} to PIR codes and codes with the disjoint repair group property.
Our results improve the constructions of these codes based on ordinary multiplicity codes~\cite{asi2018nearly}.
First, recall that a systematically encoded linear code with the DRGP property directly gives a PIR code, i.e., for linear codes the DRGP property is strictly stronger than that of PIR codes. The linear codes constructed from lifted multiplicity codes in the following have the DGRP property, but as the focus here are PIR codes, we state the results for this code class.

First, let us recall a known result for recovering the evaluation $f^{(<s)}(\w_0)$ for an arbitrary polynomial.

\begin{lemma}[Follows from~{\cite[Theorem~14]{asi2018nearly}}]\label{lem::good direction and multivariate polynomial}
	 Let $f(\X)\in\F_q[\X]$ and a line $L$ be parameterized as $\w_0+T\v$. Define $g_\v(T):=f|_L=f(\w_0+T\v)$. Let a family of sets $Q_2, \ldots, Q_{m}$, $Q_i\subset \F_q$, $|Q_i|=s$, be given. If for all directions of the form $\v=(1, v_2, \ldots, v_{m})$, $v_i\in Q_i$, and all $0\leq j< s$, values $g_{\v}^{(j)}(0)$ are known,  then it is possible to reconstruct $f^{(<s)}(\w_0)$.
\end{lemma}

Next we prove that lifted multiplicity codes satisfy the definition of $k$-PIR codes for appropriate $k$.
\begin{theorem}[Lifted multiplicity codes are PIR codes]\label{th:: nonbinary disjoint repair group code}
  Let $q$ and $s$ be powers of two  and $m\ge 2$ be an integer such that $m\le s \le q$. The $[m,s,qs-s,q]$ lifted multiplicity code is a $k$-PIR code for $k=(q/s)^{m-1}$.
\end{theorem}
\begin{proof}%
For any line $L$ parameterized by $\w_0+T\v$ and a polynomial $f$ producing a codeword of the $[m,s,qs-s,q]$ lifted multiplicity code, the polynomial $g_\v(T):=f|_L$ is equivalent up to order $s$ to a univariate polynomial $h(T)$ of degree at most $qs-s-1$. By reading $g_\v^{(j)}(t)$ for all $0\le j<s$, $t\in\F_q\setminus\{0\}$, we can reconstruct polynomial $h(T)$ in $O(qs\log(qs))$ time (cf.~\cite{chin1976generalized}) and get the values $h^{(j)}(0)=g_\v^{(j)}(0)$ for all $0\le j<s$.

For an integer $i\in [q/s]$, let $Q_i$ be a subset of $\F_q$ of size $s$ so that $Q_i\cap Q_j=\emptyset$ for $j\neq i$. Let us index codeword symbols by elements of $\F_q^m$, i.e., $(c_1,\ldots,c_{q^m})=(c_{\a})|_{\a\in \F_q^m}$, where $c_{\a}:=f^{(<s)}(\a)$. Fix an arbitrary vector $(i_2,\ldots, i_{m})\in [q/s]^{m-1}$. By Lemma~\ref{lem::good direction and multivariate polynomial}, for $\w_0\in\F_q^m$, a possible recovering set for $c_{\w_0}$ is simply
\begin{align*}
	\left\{\w_0+\v t:\, t\in \F_q\setminus\{0\},\,v_1=1,\,v_j\in Q_{i_j}\text{ for }j\in[m]\setminus\{1\}\right\}.
\end{align*}
Thus, for $c_{\w_0}$, we can construct at least $(q/s)^{m-1}$ mutually disjoint recovering sets.
\end{proof}
\begin{theorem}[Non-binary PIR codes]\label{th::asymptotic non-binary disjoint repair group code}
  Given an integer $m\ge 2$, for any real $\varepsilon$ with $0<\varepsilon <\frac{m-1}{m}$ and a power of two $q$, there exists an  $n^\varepsilon$-PIR code of length $N=q^m$ and dimension $n$ over $\Sigma$
  such that the redundancy, $N-n$,   and the alphabet size, $|\Sigma|$, satisfy
  \begin{align*}
    N-n &= O_m\left(n^{\frac{m-1}{m}+\frac{(1+\log \lambda_m - m)}{m-1}\epsilon}\right),\\
    |\Sigma|&=	q^{\Theta_m(q^{m-\epsilon m^2/(m-1)})}.
  \end{align*}
  In other words, for $0<\epsilon<1$, the polynomial growth of the minimal redundancy of $n^\epsilon$-PIR codes with dimension $n$ is
  \begin{align*}
    \log_n\!\left(r_{|\Sigma|}(n,n^\epsilon)\right)\!\le\! \min_{m \ge \lceil 1/ (1-\epsilon) \rceil}\!\left( \frac{m\!-\!1}{m} \!+\! \frac{1\!+\!\log \lambda_m\! -\! m}{m-1}\epsilon \right).
  \end{align*}
\end{theorem}
\begin{proof}%
Take $s=\Theta_m(q^{1 - \epsilon m/(m-1)})$. For simplicity of notation, we assume that $s$ is a power of two. By Theorem~\ref{th:: nonbinary disjoint repair group code}, there exists a $k$-PIR code with $k=(q/s)^{m-1}=\Theta_m(N^\varepsilon) = \Theta_m(n^{\epsilon})$ over $\F_q^{\binom{s+m-1}{m}}$ of length $N=q^m$ and redundancy at most
\begin{align*}
  N-n&=O_m\left(q^m s^{-1}(q/s)^{\log\lambda_m - m}\right)\\
     &=O_m\left(q^{\epsilon \frac{m}{m-1}+(m-1)}q^{\epsilon \frac{m}{m-1} (\log\lambda_m - m)}\right)\\
     &=O_m\left(n^{\frac{m-1}{m}+\frac{1+\log \lambda_m - m)}{m-1}\epsilon}\right) .
\end{align*}
\end{proof}

We now transform the non-binary codes constructed in Theorem~\ref{th::asymptotic non-binary disjoint repair group code} into binary PIR codes.
\begin{theorem}[Binary PIR codes]\label{th::binary disjoint repair group property code}
  Given a positive integer $m$, for any real $\varepsilon$ with $0<  \varepsilon <\frac{m-1}{m}$ and an integer $n$ sufficiently large,
  there exists a binary $n^{\varepsilon}$-PIR code of length $N$ and dimension $n$
  such that the redundancy, $N-n$, satisfies
  \begin{align*}
    N-n = n^{\frac{2m-1}{2m} + \frac{1 + 2\log \lambda_m - 2m}{2m-2}\epsilon +o_m(1)}.
  \end{align*}
  In other words, for $0<\epsilon<1$, the polynomial growth of the minimal redundancy of binary $n^\epsilon$-PIR codes with dimension $n$ is
  \begin{align*}
    \log_n\left(r(n,n^\epsilon)\right)\!\le\! \min_{m \ge \lceil 1/ (1-\epsilon) \rceil}\!\left( \frac{2m\!-\!1}{2m}\! +\! \frac{1\!+\!2\log \lambda_m\! -\! 2m}{2m-2} \epsilon\! \right).
  \end{align*}
\end{theorem}
\begin{proof}
  Let $\C$ be a non-binary PIR code as in Theorem~\ref{th::asymptotic non-binary disjoint repair group code}. We construct the binary PIR code $\overline{\C}$ from $\C$ by converting each symbol of the alphabet of size $|\Sigma|=q^{\Theta_m(q^{m - \epsilon m^2/(m-1)})}$ to
  \begin{align*}
    \log|\Sigma|&=\Theta_m\left(q^{m - \epsilon \frac{m^2}{m-1}}\log q\right)\\
                &=\Theta_m\left(N^{1-\epsilon \frac{m}{m-1}}\log N\right)\\
    &=\Theta_m\left(n^{1-\epsilon \frac{m}{m-1}}\log n\right)
  \end{align*}
  bits. Denote the length and the dimension of the binary code by $\overline{N}$ and $\overline{n}$, respectively. Note that for any $\delta>0$ and sufficiently large $n$, we have $\log n < n^{\delta}$. Thus, $\overline{n}= n^{2+o_m(1)-\epsilon m/(m-1)}$ and $\overline{N}=n^{2+o_m(1)-\epsilon m/(m-1)}$. Therefore, $n={\overline{n}}^{(m-1)/(2m-2-\epsilon m)+o_m(1)}$. Denote by $\overline{r} = \overline{N}-\overline{n}=(N-n)\log|\Sigma|$ the redundancy and by $\overline{k}$ the availability parameter of the new code.

  First, we note that the availability parameter of $\overline{\C}$ is at least that of $\C$. Indeed, we know that each bit in $\overline{\C}$ is a bit among $\log |\Sigma|$ bits representing some symbol in $\C$. For each recovering set of a symbol in $\C$, we get a corresponding recovering set for any bit from the image of this symbol in $\overline{\C}$. Therefore, $\overline{k}\ge k=n^\epsilon\ge {\overline{n}}^{\epsilon(m-1)/(2m-2-\epsilon m)+o_m(1)}$. Define $\overline{\epsilon}:= \epsilon(m-1)/(2m-2-\epsilon m)$. Then $\overline{k}=\overline{n}^{\overline{\epsilon}+o_m(1)}$ and $\epsilon = (2m-2)\overline{\epsilon}/(m-1+\overline{\epsilon} m)$

  Second, we rewrite the redundancy $\overline{r}$ in terms of $\overline{n}$ and $\overline{\epsilon}$ as
  \begin{align*}
    \overline{r}&=\overline{N}-\overline{n}\\
    &=n^{\frac{m-1}{m}+\frac{(\log \lambda_m - m + 1)\epsilon}{m-1}} n^{1-\frac{\epsilon m}{m-1}+o_m(1)} \\
                &=n^{\frac{2m-1}{m} + \frac{(\log \lambda_m - 2m + 1)\epsilon}{m-1}+o_m(1)} \\
                &=\overline{n}^{\frac{(m-1)(2m-1)}{2m^2 - 2m - 2\epsilon m^2} + \frac{(\log \lambda_m - 2m + 1)\epsilon}{2m-2-\epsilon m}+o_m(1)}  \\
                &=\overline{n}^{\frac{(m-1/2)}{m} + \frac{1/2 + \log \lambda_m - m}{m-1} \overline{\epsilon} +o_m(1)}.
  \end{align*}
\end{proof}

\section{Batch codes}\label{ss:: batch codes}

In this section we apply bounds on the rate of lifted RS codes from Section~\ref{ss::lifted RS codes} to obtain a new construction of batch codes with improved redundancy. Additionally, using an idea from~\cite{asi2018nearly}, we apply our results on PIR codes from Section~\ref{ss:: PIR codes} to obtain bounds on the redundancy of batch codes.

\subsection{Preliminaries and prior work}
By definition, PIR codes provide $k$ non-intersecting recovery sets for \emph{any single} information symbol. Batch codes generalize this property by requiring that \emph{any $k$-tuple} of information symbols (with repetition) can be recovered from non-intersecting subsets of codeword symbols.  Batch codes were originally motivated by different applications such as load-balancing in storage and cryptographic protocols~\cite{ishai2004batch}.  In this work, we consider a special notion of batch codes, namely \emph{primitive multiset batch codes}. For a more general study on the different notions of batch codes the reader is referred to \cite{skachek2018batch}. Formally, this class of codes is defined as follows.
\begin{defn}[Batch code,~%
  \cite{ishai2004batch}]\label{def::batch code}
  Let $F:\,\Sigma^n\to \Sigma^N$ be a map that encodes a string $x_1,\dots,x_n$ to $c_1,\dots, c_N$ and $\C$ be the image of $F$.
  The code $\C$ will be called a \textit{$k$-batch code} (or \textit{$[N,n,k]_{|\Sigma|}^{B}$ code}) over the alphabet $\Sigma$ if for every multiset of symbols $\{x_{i_1},\dots, x_{i_k} \},$ ${i_j}\in	[n]$, there exist $k$ mutually disjoint sets $R_{1},\dots , R_{k}\subset[N]$ (referred to as \textit{recovering sets}) and functions $g_1,\dots, g_k$ such that for all $\c\in\C$ and for all $j\in[k]$, $g_j(\c|_{R_{j}})=x_{i_j}$.
\end{defn}

 Several explicit and non-explicit constructions of these codes have been proposed, employing methods based on generalizations of Reed-Muller (RM) codes \cite{ishai2004batch,polyanskaya2020binary}, unbalanced expanders \cite{ishai2004batch}, graph theory \cite{rawat2016batch}, array and multiplicity codes \cite{asi2018nearly}, bi-variate lifted multiplicity codes and finite geometries~\cite{polyanskaya2020binary}. For large $k=\Omega(n)$, batch codes are closely related to constant-query \textit{locally correctable codes} and it is known~\cite{katz2000efficiency,woodruff2012quadratic} that their rate approaches zero. On the other hand, when $k=O(1)$ is fixed, there exist explicit batch code constructions with the code rate very close to one~\cite{vardy2016constructions}.

Denote by $\NB_q(n,k)$ the smallest $N$ such that there exists an $[N,n,k]_q^B$ code. Because of the above motivation, we classify batch codes by the required redundancy $\rB_q(n,k)\eqdef \NB_q(n,k)-n$. In this paper, we will be concerned with the regime of sublinear $k$, i.e., $k=n^\epsilon$ with ${n \to \infty}$ and $0\le\epsilon\le 1$. For $q=2$, we remove $q$ in the subsequent notations. Several achievability results, i.e., upper bounds on the smallest achievable $\rB_q(n,k)$, have been shown. We summarize the best presently known results that provide the smallest $\rB_q(n,n^\epsilon)$ for both binary and non-binary batch codes in the following statements. We note that the alphabet size of some constructions in Lemma~\ref{lem::knownBatch non-binary} is very large, e.g., $q=n^{\Omega(n)}$.
 \begin{lemma}\label{lem::knownBatch non-binary}
   The redundancy of non-binary batch codes satisfies:
 \begin{enumerate}
 \item \emph{Array construction~\cite{asi2018nearly,vardy2016constructions}:}\\ $\rB_q(n,k)= O_k (\sqrt{n})$ for linear batch codes with fixed $k$, $3\le k\le 5$.
    \item \emph{Lifted mult. codes and mult. codes with $m\!=\!2$~\cite{polyanskaya2020binary,asi2018nearly}:}\\ $\rB_q(n,n^{\epsilon}) = O(n^{3/4+\epsilon/2})$ for $0<\epsilon<\frac{1}{2}$. \label{item:qaryBatchLMC2}  
    \item \emph{Mult. codes~\cite{asi2018nearly}:}\\ $\rB_q(n,n^{\epsilon}) =  O(n^{\delta(\epsilon)})$ for $0\leq \epsilon \leq 1$, where $\delta(\epsilon) = \min\limits_{m > \frac{2}{1-\epsilon}} \left\{ 1- \frac{1}{m} +\frac{1+\epsilon}{2m-2} \right\}$. \label{item:qaryBatchMultOld}
    \item \emph{Finite geometry design~\cite{polyanskaya2020binary}:}\\ $\rB_q(n,n^{\epsilon}) = O(n^{\frac{3\epsilon+1}{2}}\log n)$ for $0<\epsilon<1/3$. \label{item:qaryBatchFG}
\end{enumerate}
 \end{lemma}
 Non-binary batch codes obtained from lifted RS codes, as given in Theorem~\ref{th::asymptotic non-binary batch code}, require the minimal redundancy among all known non-binary $n^{\epsilon}$-batch codes in the range $\epsilon \in (0.432,0.582]$ and those obtained from PIR codes, as given in Theorem~\ref{th::asymptotic non-binary batch code from PIR codes}, are best for $\epsilon\in [0.582,1)$. For a more detailed comparison, see Table~\ref{tab:comparisonNonBinaryBatch} in Appendix~\ref{app:comparison}.
 
 \begin{lemma}\label{lem::knownBatch}
   The redundancy of binary batch codes satisfies:
 \begin{enumerate}
   \item \emph{Array construction~\cite{asi2018nearly,vardy2016constructions}:}\\ $\rB(n,k)=O_k(\sqrt{n})$ for linear batch codes with fixed $k$, $3\le k\le 5$.
    \item \emph{Binary image of  lifted mult. codes with $m=2$~\cite{polyanskaya2020binary}:}\\ $\rB(n,n^{\epsilon}) = O(n^{\log_4(3)+(2-\log_2(3))\epsilon} \log n)$ for $0\!<\!\epsilon\!<\!\frac{1}{2}$. \label{item:binBatchLMC2}
    \item \emph{Binary image of mult. codes~\cite{asi2018nearly}:}\\ $\rB(n,n^{\epsilon}) =  O(n^{\delta(\epsilon)}\log n)$ for $0\leq \epsilon \leq 1$, where $\delta(\epsilon) = \min\limits_{m > \frac{2}{1-\epsilon}} \left\{ 1- \frac{m(1-\epsilon)-2}{4m(m-1)} \right\}$. \label{item:binBatchMultOld}
    \item \emph{Finite geometry design~\cite{polyanskaya2020binary}:}\\ $\rB(n,n^{\epsilon}) = O(n^{\frac{3\epsilon+1}{2}}\log n)$ for  $0<\epsilon<1/3$. \label{item:binBatchFG}
\end{enumerate}
 \end{lemma}
 Binary batch codes obtained from lifted RS codes, as given in Theorem~\ref{th::binary batch code}, require the minimal redundancy among all known binary $n^{\epsilon}$-batch codes in the range $\epsilon \in (0.41,0.648]$ and those obtained from PIR codes, as given in Theorem~\ref{th::binary batch code from PIR codes}, are best for $\epsilon\in [0.648,1)$. A more detailed comparison is given in Table~\ref{tab:comparisonBinaryBatch} in Appendix~\ref{app:comparison}.

On the other hand, the only non-trivial converse bound on the redundancy of systematic linear batch codes, yielding that $\rB_q(n,k)=\Omega(\sqrt{nk})$, was recently shown in~\cite{li2021improved}. %
\begin{figure*}
  \centering

\begin{tikzpicture}[thick,scale=1]
	\pgfplotsset{compat = 1.3}
	\begin{axis}[
		legend cell align={left},
		width = 0.9\textwidth,
		height = 0.63\textwidth,
		xlabel = {$\log_n(k)$},
		ylabel = {$\log_n(\rB_q(n,k))$},
		xmin = 0,
		xmax = 1,
		ymin = 0.48,
		ymax = 1.03,
		legend pos = south east]

		\addplot [loosely dotted, very thick, color=red, mark=none] table[x=eps,y=bin] {upperBoundBatchFile.txt};
		\addlegendentry{Prior work (binary, Lem.~\ref{lem::knownBatch})}%

		\addplot[color= blue, mark=none,dotted,very thick] table[x=eps,y=binNew] {upperBoundBatchFile.txt};
		\addlegendentry{Binary image of lifted RS codes (Th.~\ref{th::binary batch code})}
		
		\addplot[color= cyan, mark=none,densely dotted,very thick] table[x=eps,y=binNewPIR] {upperBoundBatchFile.txt};
		\addlegendentry{Binary image of lifted mult. codes (Th.~\ref{th::binary batch code from PIR codes})}

		\addplot [loosely dashed, very thick, color=purple, mark=none] table[x=eps,y=qary] {upperBoundBatchFile.txt};
		\addlegendentry{Prior work (non-binary, Lem.~\ref{lem::knownBatch non-binary})}%

		\addplot[color=OliveGreen, mark=none, dashed, very thick] table[x=eps,y=qaryNew] {upperBoundBatchFile.txt};
		\addlegendentry{Lifted RS codes (Th.~\ref{th::asymptotic non-binary batch code})}
		
		\addplot[color=green, mark=none,densely dashed, very thick] table[x=eps,y=qaryNewPIR] {upperBoundBatchFile.txt};
		\addlegendentry{Lifted mult. codes (Th.~\ref{th::asymptotic non-binary batch code from PIR codes})}
		
		\addplot[color=black,mark =none,thick] coordinates { (0,0.5) (1,1) };
		\addlegendentry{Lower bound~\cite{li2021improved}}
	\end{axis}
\end{tikzpicture}

  \caption{Comparison of bounds on the parameters of batch codes based on $m$-variate lifted RS and lifted multiplicity codes for different values of $m$ to the upper and lower bounds of \cite{wootters2016linear,rao2016lower,asi2018nearly,vardy2016constructions,polyanskaya2020binary}.
  }
  \label{fig:batch}
\end{figure*}
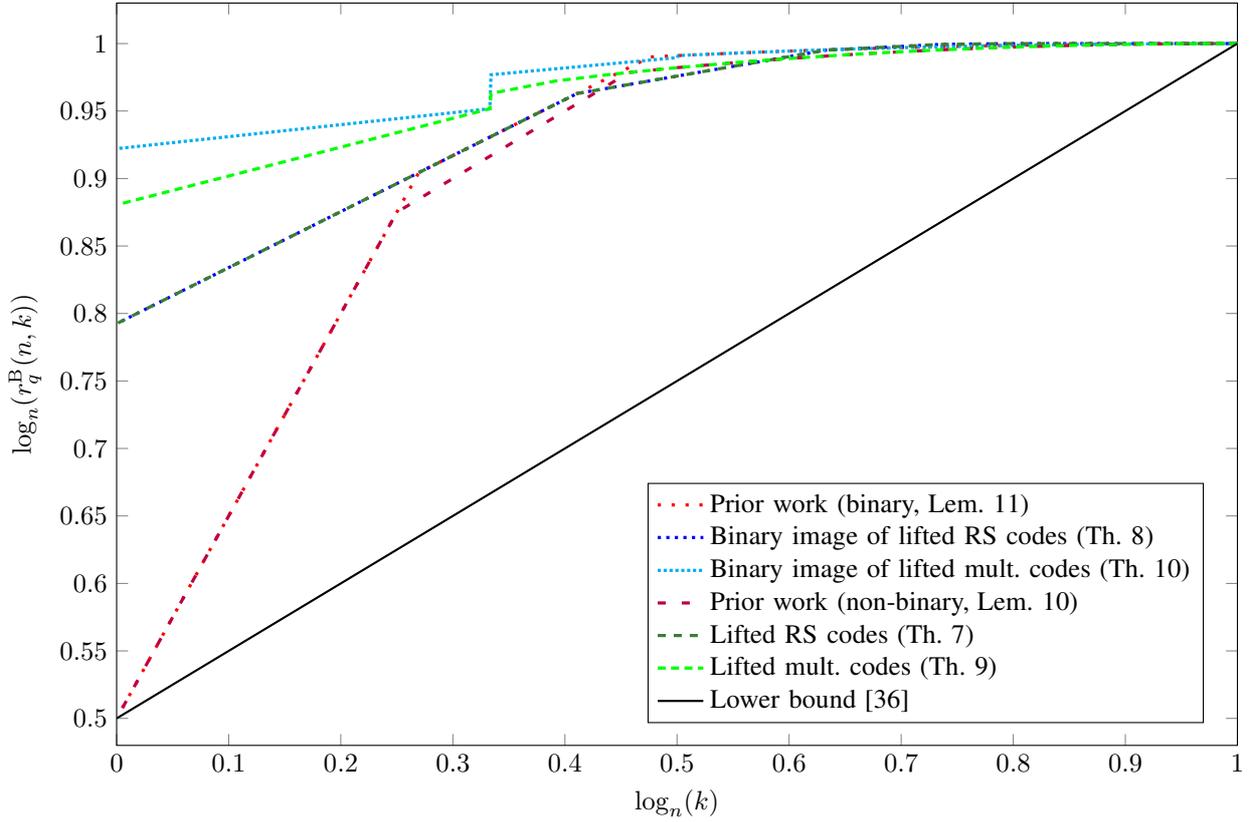

We illustrate the trade-off between parameters of batch codes in Figure~\ref{fig:batch}.

\subsection{Batch codes from lifted RS codes}
In this section, a new construction of binary batch codes is presented.
 To this end, we first provide a construction of non-binary $k$-batch codes of length $n$ based on the $m$-dimensional lift of an RS code. After that, we compute the parameters of this construction in the asymptotic regime for the availability parameter $k=n^{\varepsilon}$ with real $\varepsilon\in[\frac{m-2}{m}, \frac{m-1}{m}]$. Finally, we show how to convert this construction into a binary batch code.

We now provide a one-way connection between lifted RS codes and batch codes.%

\begin{theorem}\label{th:: nonbinary batch code}
Fix integers $q$, $m$ and $r<q$. The $[m,q-r,q]$ lifted RS code has the following properties:
\begin{enumerate}
\item The length of the code is $q^m$.
\item The rate of the code is $1-\Theta\left((q/r)^{\log \lambda_m-m}\right)$ as $q\to\infty$.
\item The code is a $k$-batch code for $k=q^{m-2}r$.
\end{enumerate}
\end{theorem}
\begin{proof}[Proof of Theorem~\ref{th:: nonbinary batch code}]
The first property follows from Definition~\ref{def::lifted RS code}. The second property is implied by Theorem~\ref{th:: number of bad monomials}.

To prove the third property, we first note that a lifted RS code is a linear code over $\F_q$ and it can be encoded systematically. Let $\y$ be a codeword of the $[m,d,q]$ lifted RS code. Since every coordinate of $\y$ is simply the evaluation $f(\a)$ for some $\a\in\F_q^m$, we can index coordinates of our code by elements $\a$ from $\F_q^m$.

Now we shall prove a slightly stricter condition than required for $k$-batch codes, namely for every multiset of codeword symbols $\{y_{\a_1},\dots,y_{\a_k}\}$, there exist mutually disjoint sets $R_1,\dots, R_k\subset \F_q^m$ and some functions $g_1,\dots,g_k$ such that $y_{\a_i}=g_i(\y|_{R_i})$.
Let us prove the existence of $R_1,\dots,R_k$ by using the inductive procedure described below.

To reconstruct $y_{\a_1}$, we take an arbitrary line $L_1$ in $\F_q^m$ containing $\a_1$ and let $R_1=L_1\setminus \{\a_1\}$. As the restriction of polynomial $f$ to a line $L_1$ has degree less than $q-r$ by definition of lifted RS codes, we can interpolate $f|_{L_1}$ by reading evaluations of $f$ at some $q-r$ points on the line $L_1$ and evaluate $f|_{L_1}$ at point $\a_1$. Suppose that for $k'<k$, symbols $\{y_{\a_1},\dots,y_{\a_{k'}}\}$ can be already reconstructed by using recovering sets $R_1,\dots, R_{k'}$, where $R_i$ is a subset of a line $L_i$ from the space $\F_q^m$. Since the number of lines passing through the point $\a_{k'+1}$ is larger than $q^{m-1}$ and the total number of points already employed for recovering $\{y_{\a_1},\dots,y_{\a_{k'}}\}$ is at most $qk'$, we conclude that there exists a line $L_{k'+1}$ among $q^{m-1}$ ones such that the cardinality of the intersection
\begin{align*}
\left|L_{k'+1} \bigcap \left\{\bigcup_{i\in[k']}L_{i}\right\}\right|\le \frac{qk'}{q^{m-1}}< \frac{qk}{q^{m-1}}=r.
\end{align*}
Therefore, we can reconstruct $\y_{\a_{k'+1}}$ by reading evaluations of $f$ at some $q-r$ unused points on $L_{k'+1}$, interpolating the univariate polynomial $f|_{L_{k'+1}}$ of degree less than $q-r$ and evaluating the latter at point $\a_{k'+1}$.

Thus, the required multiset of codeword symbols can be determined by this procedure. This completes the proof.
\end{proof}
In the next statement we show a connection between parameters of the non-binary batch code constructed in Theorem~\ref{th:: nonbinary batch code}.
\begin{theorem}\label{th::asymptotic non-binary batch code}
Given a positive integer $m$, for any real $\varepsilon$ with $\frac{m-2}{m}<\varepsilon <\frac{m-1}{m}$ and a sufficiently large power of two $q$,
there exists a $n^\varepsilon$-batch code of length $N=q^m$ and dimension $n$ over $\F_q$
such that the redundancy, $N-n$,   satisfies
\begin{align*}
N-n = O_m\left(n^{(m-\log \lambda_m)\varepsilon + \frac{(m-1)\log \lambda_m}{m} - m + 2}\right).
\end{align*}
\end{theorem}
\begin{proof}[Proof of Theorem~\ref{th::asymptotic non-binary batch code}]
Let $r= \lceil q^{m\varepsilon-m+2}\rceil \ge n^{\varepsilon-(m-2)/m}$. By Theorem~\ref{th:: nonbinary batch code}, there exists a $k$-batch code with $k=rq^{m-2}\ge q^{m\varepsilon}\ge n^\varepsilon$ over $\F_q$ of length $N=q^m$ and redundancy at most
\begin{align*}
N-n&=O_m\left(r^m \lambda_m^{\ell-\log r}\right)\\
&=O_m\left(2^{\ell m (m\varepsilon-m+2)}\lambda_m^{\ell-\ell (m\varepsilon-m+2)}\right)\\
&=O_m\left(n^{(m-\log \lambda_m)\varepsilon + \frac{(m-1)\log \lambda_m}{m} - m + 2}\right) .
\end{align*}
\end{proof}
\begin{theorem}\label{th::binary batch code}
	Given a positive integer $m$, for any real $\varepsilon$ with $\frac{m-2}{m}<  \varepsilon <\frac{m-1}{m}$ and an integer $n$ sufficiently large,
there exists a binary $n^{\varepsilon}$-batch code of length $N$ and dimension $n$
such that the redundancy, $N-n$, satisfies
\begin{align*}
N-n = n^{(m-\log \lambda_m)\varepsilon + \frac{(m-1)\log \lambda_m}{m} - m + 2+o_m(1)}.
\end{align*}
\end{theorem}
\begin{proof}[Proof of Theorem~\ref{th::binary batch code}]
	Let $\C$ be a non-binary batch code from Theorem~\ref{th::asymptotic non-binary batch code}. We construct the binary batch code $\C'$ from $\C$ by converting each symbol of the alphabet of size $q$ to $\log q=\log N^{1/m}=\frac{1}{m}\log N=\Theta(\log n)$ bits. Denote the length, dimension of the binary
	code by $N', n'$ respectively. Thus, $n'= \Theta(n \log n)$ and $N'=\Theta(N \log n)$. Therefore, $n=\Theta(n'/\log n')$. Denote by $r' = N'-n'$ the redundancy of the binary code and by $k'$ be the availability parameter of the new code.

	First, we note that the availability parameter of $\C'$ is at least that of $\C$. Indeed, we know that each bit in $\C'$ is a bit among $\log q$ bits representing some symbol in $\C$. For each recovering set of a symbol in $\C$, we have the corresponding recovering set for any bit from the image of this symbol in $\C'$. Therefore, $k'\ge k=n^\epsilon\ge (n'/\log n')^{\varepsilon}$.

	Second, we rewrite the redundancy $r'$ in terms of $n'$ as
	\begin{align*}
	r'&=N'-n'=O((N-n)\log n) \\
	&=O_m\left(n'^{(m-\log \lambda_m)\varepsilon + \frac{(m-1)\log \lambda_m}{m} - m + 2}\log n'\right) .
	\end{align*}
As for any $\delta>0$ and sufficiently large $n$ we have $\log n < n^{\delta}$, the required statement is proved.
\end{proof}

\subsection{Batch codes from PIR codes}
Batch codes from bi-variate lifted multiplicity codes were derived in~\cite{polyanskaya2020binary}, however, it is difficult to obtain such results from lifted multiplicity codes with a larger number of variables (for a similar discussion about batch codes from multiplicity codes, we refer the reader to~\cite{asi2018nearly}). However, by the generic connection between PIR and batch codes we are able to indirectly construct batch code from lifted multiplicity codes.

Recall the result from~\cite[Theorem 30]{asi2018nearly} which relates the redundancy of batch and PIR codes
$$
\rB_q(n,n^{\epsilon})\le \rP_q(n,n^{\frac{1+\epsilon}{2}}).
$$
Combining this bound with Theorems~\ref{th::asymptotic non-binary disjoint repair group code} and~\ref{th::binary disjoint repair group property code} yields the following statements.
\begin{theorem}[Non-binary batch codes]\label{th::asymptotic non-binary batch code from PIR codes}
  Given an integer $m\ge 3$, for any real $\varepsilon$ with $0<\varepsilon <\frac{m-2}{m}$ and a power of two $q$, there exists an  $n^\varepsilon$-batch code of length $N=q^m$ and dimension $n$ over $\Sigma$
  such that the redundancy, $N-n$,   and the alphabet size, $|\Sigma|$, satisfy
  \begin{align*}
    N-n &= O_m\left(n^{\frac{m-1}{m} + \frac{(1+\log\lambda_m -m)(1+\epsilon)}{2m-2}}\right),\\
    |\Sigma|&=	q^{\Theta_m(q^{m-\epsilon m^2/(m-1)})}.
  \end{align*}
  In other words, for $0<\epsilon<1$, the polynomial growth of the minimal redundancy of $n^\epsilon$-batch codes with dimension $n$ is
  \begin{align*}
    \log_n&\left(\rB_{|\Sigma|}(n,n^\epsilon)\right)\\
    &\le \min_{m \ge \lceil 2/ (1-\epsilon) \rceil}\left( \frac{m-1}{m} + \frac{(1+\log \lambda_m - m)(1+\epsilon)}{2m-2} \right).
  \end{align*}
\end{theorem}
\begin{theorem}[Binary batch codes]\label{th::binary batch code from PIR codes}
  Given a positive integer $m$, for any real $\varepsilon$ with $0<  \varepsilon <\frac{m-2}{m}$ and an integer $n$ sufficiently large,
  there exists a binary $n^{\varepsilon}$-batch code of length $N$ and dimension $n$
  such that the redundancy, $N-n$, satisfies
  \begin{align*}
    N-n = n^{\frac{2m-1}{2m} + \epsilon \frac{(1 + 2\log \lambda_m - 2m)(1+\epsilon)}{4m-4}+o_m(1)}
  \end{align*}
  In other words, for $0<\epsilon<1$, the polynomial growth of the minimal redundancy of binary $n^\epsilon$-batch codes with dimension $n$ is
  \begin{align*}
    &\log_n\left(\rB(n,n^\epsilon)\right)\\
    &\ \ \le \min_{m \ge \lceil 2/ (1-\epsilon) \rceil}\left( \frac{2m-1}{2m} + \frac{(1+2\log \lambda_m - 2m)(1+\epsilon)}{4m-4} \right).
  \end{align*}
\end{theorem}

\section{Locally correctable codes}\label{ss:: LCCs}
In this section we show that a lifted multiplicity code is a locally correctable code (LCC) with certain parameters. Specifically, we provide a self-correction algorithm for lifted multiplicity codes.
\subsection{Preliminaries and prior work}
Unlike PIR codes, LCCs \cite{katz2000efficiency} explicitly require locality properties. Informally, a code is said to be locally correctable if given a vector that is sufficiently close to a codeword, each codeword coordinate can be recovered from a small subset of (possibly noisy) other positions with high probability. We give a formal definition of LCCs below.
\begin{defn}[Locally correctable code.] \label{def::LCC}
    A code $\C$ of length $N$ over an alphabet $\Sigma$ is said to be $(r, \delta, \xi)$-locally correctable if there exists a randomized correcting algorithm $\A$ such that
    \begin{enumerate}
    \item For all $\c\in \C$, $i\in [N]$ and all vectors $\y \in \Sigma^N$ such that the relative distance $\Delta(\y, \c)\leq \delta $, we have
      \begin{align*}
\Pr (\A(\y, i)=c_i)\geq 1-\xi.
\end{align*}
        \item $\A$ makes at most $r$ queries to $\y$.
    \end{enumerate}
\end{defn}
LDCs \cite{yekhanin2012locally} are defined similar to LCCs, except that there the algorithm is required to recover message symbols instead of codeword symbols. Note, that for linear codes local correctability is a strictly stronger notion than local decodability, as a systematically encoded LCC is always an LDC.

LCCs have been constructed employing different approaches such as RM codes, lifted RS codes~\cite{guo2013new}, multiplicity codes~\cite{kopparty2014high}, and tensor codes~\cite{ben2006LCC,viderman2015combination}.  One typical question about LCCs is phrased as follows: given the high rate of a code (close to 1), how to get the query complexity as small as possible. The current state-of-the-art construction provided in~\cite{kopparty2017high} has the sub-polynomial (in length) query complexity.
For an extensive discussion about other aspects of LCCs see~\cite{trevisan2004some, yekhanin2012locally, kopparty2017local} and the references therein.

\subsection{LCCs from lifted multiplicity codes} \label{ss::locally correctable codes}

One important ingredient to show the self-correction algorithm for lifted multiplicity codes is the following statement about hypergraphs. Recall that an \textit{$s$-partite hypergraph} $H$ is a pair $H=(V,E)$, where $V$ is the vertex set that can be partitioned into sets $V_1,\ldots,V_s$ so that each edge in the edge set $E$ consists of a choice of precisely one vertex from each part. By $K_l^{(s)}$ denote a \textit{complete $s$-partite hypergraph}, whose parts are all of equal size $l$.

\begin{theorem}[Follows from~{\cite[Theorem~1]{erdos1964extremal}}] \label{th::Turan's type theorem}
 		Let $n>sl$, $l>1$.  Then every $s$-partite hypergraph with $n$ vertexes and at least $n^{s-1/l^{s-1}}$ hyperedges contains a copy of $K^{(s)}_l$.
 \end{theorem}

\begin{theorem}\label{th:: self correction algorithm}
Let $m$ be a fixed positive integer.    For $s^{m-2}=o(\log q)$ and a real $\alpha<1/4$, the $[m, s, qs-r, q]$ lifted multiplicity code is a $((q-1)s^{m-1}, \alpha \Delta_{min}, 2\alpha +o(1))$-locally correctable code, where $\Delta_{min}:=\left \lceil\frac{r-s+1}{s}\right\rceil\frac{q-s}{q^2}$.
\end{theorem}
\begin{remark}
It is worth mentioning that the self-correction algorithm for multiplicity codes from~\cite{kopparty2014high}, which  has the query complexity $(q-1)5^m (s+1)^m$, also works well for lifted multiplicity codes.  In the algorithm proposed in the proof of Theorem~\ref{th:: self correction algorithm}, we impose a stronger requirement on the order of derivatives: $s^{m-2}=o(\log q)$ for our algorithm and $s\le q/5-1$ for the algorithm from~\cite{kopparty2014high}. However, our proposed algorithm has the query complexity $(q-1)s^{m-1}$, which implies a slightly better running time. For instance, the complexity  of our algorithm is $\Theta_m(s)$ times smaller when $m$ is fixed, $q\to\infty$ and $s=(\log q)^{1/(m-1)}$.
\end{remark}
\begin{proof}
  We prove this theorem by presenting a new self-correction algorithm $\A$ for lifted multiplicity codes. Consider a vector $\y=(y_1, \ldots, y_{q^m})=(y_{\a})|_{\a\in \F_q^m}$, which is a noisy version of the evaluation of the polynomial $f$. Say that we want to correct the value $f^{(<s)}(\w_0)$ with some $\w_0\in\F_q^m$.  The algorithm $\A$ consists of three steps. \medskip

  \textbf{Step 1:}
  Choose sets $Q_2, Q_3, \ldots, Q_m$, $Q_i\subset \F_q$, $|Q_i|=s$, independently according to the uniform distribution over all subsets of size $s$. Form a set $V$ of directions $\v=(1, v_2, \ldots, v_m), v_i\in Q_i$. Clearly, $|V|=s^{m-1}$.\medskip

  \textbf{Step 2:}
  For every $\v \in V$ define a polynomial $g_{\v}(T):=f(\w_0+T\v)$. By the definition of lifted multiplicity codes, this polynomial agrees with some univariate polynomial of degree less than $qs-r$ on its first  $s-1$ derivatives. Apply the decoding algorithm for a univariate multiplicity code from~\cite{kopparty2014high, sudan2001ideal} to noisy evaluations of $g_{\v}(T)$ to obtain an estimation $\hat{g}_{\v}(T)$ of the correct polynomial $g_{\v}(T)$. Note that this decoding algorithm can correct up to $\lfloor (d_{min}-1)/2 \rfloor$ errors, where $d_{min}:=\lceil\frac{r+1}{s}\rceil$.\medskip

  \textbf{Step 3:}
  Using Lemma~\ref{lem::good direction and multivariate polynomial} and polynomials $\hat{g}_{\v}(T)$, recover the value $f^{(<s)}(\w_0)$ to obtain $\hat f^{(<s)}(\w_0)$.\medskip

  We now present an analysis of the algorithm. Call a direction $\v$ \textit{good}, if the line $\w_0+T\v$ contains at most $\lfloor (d_{min}-1)/2 \rfloor$ errors. Note that if a direction $\v$ is good, then $\hat{g}_{\v}(T)\equiv_s g_{\v}(T)$. Thus, if all directions from $V$ are good, the algorithm  recovers the symbol correctly, i.e., $\hat f^{(<s)}(\w_0)=f^{(<s)}(\w_0)$. In the following we derive a bound on the probability that all directions from $V$ are good.

  Introduce an $(m-1)$-uniform $(m-1)$-partite hypergraph $H$, each part of which has size $q$. Index the elements within each part of the hypergraph with elements of $\F_q$. For every good direction $\v=(1, v_2, \ldots, v_{m})$, draw a hyperedge $(v_2, \ldots, v_{m})$ in $H$, where $v_i$ is a vertex from the $(i-1)$th part.
  Then the probability of the successful recovery of $f^{(<s)}(\w_0)$ is lower bounded by the number of copies of $K^{(m-1)}_s$ in $H$ divided by $q^{m-1}$.

  The total number of good directions (or hyperedges in $H$) is at least
  \begin{align*}
    q^{m-1}- \frac{\alpha \Delta_{min}q^m}{\lfloor (d_{min}-1)/2 \rfloor}=q^{m-1}(1-2\alpha +o(1)).
  \end{align*}
  We show how we can find a large number of copies of $K^{(m-1)}_s$ in $H$.
 As long as the number of hyperedges in $H$ is greater than $((m-1)q)^{m-1-1/s^{m-2}}$ we can find such a copy by Theorem~\ref{th::Turan's type theorem}. Then, we can spoil this copy by erasing one of its hyperedges and repeat the process for the obtained hypergraph. Obviously, all constructed copies of $K^{(m-1)}_s$ will be distinct. By this procedure, we can find at least
 \begin{align*}
   q^{m-1}(1-2\alpha+o(1))-&((m-1)q)^{m-1-1/s^{m-2}}\\
   &\quad \qquad =q^{m-1}(1-2\alpha+o(1))
 \end{align*}
 copies of $K^{(m-1)}_s$. Therefore, the probability of successful decoding is at least $1-2\alpha+o(1)$.
\end{proof}

\section{Conclusion}\label{ss::conclusion}
In this paper, we have investigated the rate, the distance, the availability and the self-correction properties of lifted Reed-Solomon codes and lifted multiplicity codes based on the evaluations of $m$-variate polynomials and discussed how to use them to construct batch codes, PIR codes, and LCCs. For some parameter regimes, such codes obtained from lifted RS and lifted multiplicity codes are shown to have a better rate/distance/availability/locality trade-off than other known constructions. In particular, our main results are:
\begin{enumerate}[wide]
\item We have improved the estimate on the rate of the $m$-dimensional lifts of RS codes when the field size is large. In particular, we have shown that for $r=O(1)$, the $[m,q-r,q]$ lifted RS code has rate $1-\Theta(q^{\log{\lambda_m}-m})$ as $q\to\infty$.
\item We have continued the study of lifted multiplicity codes initiated for the bi-variate case in~\cite{li2019lifted} for any number of variables $m\geq 3$. Specifically, we show the rate of the $[m,s,sq-r,q]$ lifted multiplicity code to be $1-O_m(s^{-1}(q/r)^{\log \lambda_m - m})$ and its relative distance to be $\Delta_{min}=\frac{r}{qs}(1+o(1))$, by analyzing the code obtained from the span of good monomials. An interesting open problem is to extend this analysis to the code spanned by all good \emph{polynomials}, i.e., the complete lifted multiplicity code.
\item We have proved that an $[m,s,sq-s,q]$ lifted multiplicity  code is a $k$-PIR code of dimension $n=q^m(1+o(1))$ with $k=(q/s)^{m-1}$. This improves the known upper bounds on the redundancy of PIR codes when $k$ is sublinear in $n$ and $k\ge \sqrt{n}$. For small enough $s$ and any constant $\alpha<1/4$, the $[m,s,sq-s,q]$ lifted multiplicity code is shown to be a $\left(qs^{m-1}, \alpha\Delta_{min}, 2\alpha \right)$-locally correctable code.
\item We have shown that an $[m,q-r,q]$-lifted RS code is also a $k$-batch code with $k=rq^{m-2}$ and, by a generic transformation, we provide results on batch codes obtained from lifted multiplicity codes.
This improves the known upper bounds on the redundancy of batch codes in some parameter regimes. On the other hand, there is no lower bound on the redundancy of batch and PIR codes other than that for $k\ge 3$ the redundancy of linear $k$-batch and $k$-PIR codes of length $N$ is $\Omega(\sqrt{kN})$ and $\Omega(\sqrt{N})$, respectively~\cite{li2021improved,rao2016lower,wootters2016linear}. Closing the (large) gap between the lower and upper bounds on the redundancy of both batch and PIR codes remains a major open problem.
\end{enumerate}

	\section{Acknowledgment}
	The authors are grateful to the anonymous reviewers for their    careful    reading    of    the    manuscript    and    their    many    insightful    comments    and    suggestions which improved both the exposition of the paper and the clarity of the proofs.

\newpage

\bibliographystyle{IEEEtran}
\bibliography{lifted}

\begin{thebibliography}{10}
\providecommand{\url}[1]{#1}
\csname url@samestyle\endcsname
\providecommand{\newblock}{\relax}
\providecommand{\bibinfo}[2]{#2}
\providecommand{\BIBentrySTDinterwordspacing}{\spaceskip=0pt\relax}
\providecommand{\BIBentryALTinterwordstretchfactor}{4}
\providecommand{\BIBentryALTinterwordspacing}{\spaceskip=\fontdimen2\font plus
\BIBentryALTinterwordstretchfactor\fontdimen3\font minus
  \fontdimen4\font\relax}
\providecommand{\BIBforeignlanguage}[2]{{%
\expandafter\ifx\csname l@#1\endcsname\relax
\typeout{** WARNING: IEEEtran.bst: No hyphenation pattern has been}%
\typeout{** loaded for the language `#1'. Using the pattern for}%
\typeout{** the default language instead.}%
\else
\language=\csname l@#1\endcsname
\fi
#2}}
\providecommand{\BIBdecl}{\relax}
\BIBdecl

\bibitem{holzbaur2020lifted}
L.~{Holzbaur}, R.~{Polyanskaya}, N.~{Polyanskii}, and I.~{Vorobyev}, ``Lifted
  reed-solomon codes with application to batch codes,'' in \emph{2020 IEEE Int.
  Symp. Inf. Theory (ISIT)}, 2020, pp. 634--639.

\bibitem{holzbaur2021lifted}
L.~Holzbaur, R.~Polyanskaya, N.~Polyanskii, I.~Vorobyev, and E.~Yaakobi, ``On
  lifted multiplicity codes,'' in \emph{2020 IEEE Information Theory Workshop
  (ITW)}.\hskip 1em plus 0.5em minus 0.4em\relax IEEE, 2021, pp. 1--5.

\bibitem{huang2013pyramid}
C.~Huang, M.~Chen, and J.~Li, ``Pyramid codes: Flexible schemes to trade space
  for access efficiency in reliable data storage systems,'' \emph{ACM Trans.
  Storage}, vol.~9, no.~1, pp. 1--28, 2013.

\bibitem{GHSY12}
P.~Gopalan, C.~Huang, H.~Simitci, and S.~Yekhanin, ``On the locality of
  codeword symbols,'' \emph{IEEE Trans. Inf. Theor.}, vol.~58, no.~11, p.
  6925–6934, Nov. 2012.

\bibitem{katz2000efficiency}
J.~Katz and L.~Trevisan, ``On the efficiency of local decoding procedures for
  error-correcting codes,'' in \emph{Proc. 32nd Annu. ACM Symp. Theory Comput.
  (STOC)}, 2000, pp. 80--86.

\bibitem{yekhanin2012locally}
S.~Yekhanin \emph{et~al.}, ``Locally decodable codes,'' \emph{Found. Trends
  Theor. Comput. Sci.}, vol.~6, no.~3, pp. 139--255, 2012.

\bibitem{gur2018relaxed}
T.~Gur, G.~Ramnarayan, and R.~D. Rothblum, ``Relaxed locally correctable
  codes,'' in \emph{Proc. 9th Conf. Innov. Theor. Computer Sci. (ITCS)}, 2018,
  p. 27:1–27:11.

\bibitem{ben2006robust}
E.~Ben-Sasson, O.~Goldreich, P.~Harsha, M.~Sudan, and S.~Vadhan, ``Robust
  {PCP}s of proximity, shorter {PCP}s, and applications to coding,'' \emph{SIAM
  J. Comput.}, vol.~36, no.~4, pp. 889--974, 2006.

\bibitem{ishai2004batch}
Y.~Ishai, E.~Kushilevitz, R.~Ostrovsky, and A.~Sahai, ``Batch codes and their
  applications,'' in \emph{Proc. 36th Annu. ACM Symp. Theory Comput. (STOC)},
  2004, pp. 262--271.

\bibitem{fazeli2015pir}
A.~Fazeli, A.~Vardy, and E.~Yaakobi, ``{PIR} with low storage overhead: coding
  instead of replication,'' \emph{arXiv preprint arXiv:1505.06241}, 2015.

\bibitem{li2019lifted}
R.~Li and M.~Wootters, ``Lifted multiplicity codes and the disjoint repair
  group property,'' in \emph{Proc. Approx. Randomiz. Combinat. Optim. Algor.
  Techn. (APPROX/RANDOM)}, vol. 145, 2019, pp. 38:1--38:18.

\bibitem{reed1954class}
I.~Reed, ``A class of multiple-error-correcting codes and the decoding
  scheme,'' \emph{Trans. IRE Prof. Group Inf. Theory}, vol.~4, no.~4, pp.
  38--49, 1954.

\bibitem{arora2003improved}
S.~Arora and M.~Sudan, ``Improved low-degree testing and its applications,''
  \emph{Combinatorica}, vol.~23, no.~3, pp. 365--426, 2003.

\bibitem{alon2005testing}
N.~Alon, T.~Kaufman, M.~Krivelevich, S.~Litsyn, and D.~Ron, ``Testing
  {Reed-Muller} codes,'' \emph{IEEE Trans. Inf. Theory}, vol.~51, no.~11, pp.
  4032--4039, 2005.

\bibitem{rubinfeld1996robust}
R.~Rubinfeld and M.~Sudan, ``Robust characterizations of polynomials with
  applications to program testing,'' \emph{SIAM J. Comput.}, vol.~25, no.~2,
  pp. 252--271, 1996.

\bibitem{guo2013new}
A.~Guo, S.~Kopparty, and M.~Sudan, ``New affine-invariant codes from lifting,''
  in \emph{Proc. 4th Conf. Innov. Theor. Computer Sci. (ITCS)}, 2013, pp.
  529--540.

\bibitem{ben2011symmetric}
E.~Ben-Sasson, G.~Maatouk, A.~Shpilka, and M.~Sudan, ``Symmetric {LDPC} codes
  are not necessarily locally testable,'' in \emph{IEEE 26th Annu. Conf.
  Comput. Complex. (CCC)}, 2011, pp. 55--65.

\bibitem{kopparty2014high}
S.~Kopparty, S.~Saraf, and S.~Yekhanin, ``High-rate codes with sublinear-time
  decoding,'' \emph{J. Assoc. Comput. Mach.}, vol.~61, no.~5, p.~28, 2014.

\bibitem{wu2015revisiting}
L.~Wu, ``Revisiting the multiplicity codes: A new class of high-rate locally
  correctable codes,'' in \emph{Proc. IEEE 53rd Annu. Allerton Conf. Commun.
  Contr. Comput. (Allerton)}, 2015, pp. 509--513.

\bibitem{polyanskii2019lifted}
N.~Polyanskii and I.~Vorobyev, ``Trivariate lifted codes with disjoint repair
  groups,'' in \emph{Proc. IEEE XVI Int. Symp. Probl. Redund. Inf. Contr. Syst.
  (REDUNDANCY)}, 2019, pp. 64--68.

\bibitem{dvir2013extensions}
Z.~Dvir, S.~Kopparty, S.~Saraf, and M.~Sudan, ``Extensions to the method of
  multiplicities, with applications to kakeya sets and mergers,'' \emph{SIAM J.
  Comput.}, vol.~42, no.~6, pp. 2305--2328, 2013.

\bibitem{horn2012matrix}
R.~A. Horn and C.~R. Johnson, \emph{Matrix analysis}.\hskip 1em plus 0.5em
  minus 0.4em\relax Cambridge university press, 2012.

\bibitem{VRK17}
M.~{Vajha}, V.~{Ramkumar}, and P.~{Vijay Kumar}, ``Binary, shortened projective
  reed muller codes for coded private inf retrieval,'' in \emph{Proc. IEEE Int.
  Symp. Inf. Theory (ISIT)}, 2017, pp. 2648--2652.

\bibitem{rao2016lower}
S.~Rao and A.~Vardy, ``Lower bound on the redundancy of {PIR} codes,''
  \emph{arXiv preprint arXiv:1605.01869}, 2016.

\bibitem{wootters2016linear}
M.~Wootters, ``Linear codes with disjoint repair groups,'' \emph{unpublished
  mansucript, February}, 2016.

\bibitem{asi2018nearly}
H.~Asi and E.~Yaakobi, ``Nearly optimal constructions of {PIR} and batch
  codes,'' \emph{IEEE Trans. Inf. Theory}, vol.~65, no.~2, pp. 947--964, 2018.

\bibitem{FGW17}
S.~L. Frank{-}Fischer, V.~Guruswami, and M.~Wootters, ``Locality via partially
  lifted codes,'' in \emph{Proc. Approx. Randomiz. Combinat. Optim. Algor.
  Techn. (APPROX/RANDOM)}, vol.~81, 2017, pp. 43:1--43:17.

\bibitem{LC04}
S.~Lin and D.~J. Costello, \emph{Error control coding: fundamentals and
  applications}.\hskip 1em plus 0.5em minus 0.4em\relax Upper Saddle River, NJ:
  Pearson/Prentice Hall, 2004.

\bibitem{hastings2020wedge}
J.~Hastings, A.~Kanne, R.~Li, and M.~Wootters, ``Wedge-lifted codes,'' in
  \emph{Proc. IEEE Int. Symp. Inf. Theory (ISIT)}, 2021, pp. 2990--2995.

\bibitem{chin1976generalized}
F.~Y. Chin, ``A generalized asymptotic upper bound on fast polynomial
  evaluation and interpolation,'' \emph{SIAM J. Comput.}, vol.~5, no.~4, pp.
  682--690, 1976.

\bibitem{skachek2018batch}
V.~Skachek, ``Batch and {PIR} codes and their connections to locally repairable
  codes,'' in \emph{Network Coding and Subspace Designs}.\hskip 1em plus 0.5em
  minus 0.4em\relax Springer, 2018, pp. 427--442.

\bibitem{polyanskaya2020binary}
R.~Polyanskaya, N.~Polyanskii, and I.~Vorobyev, ``Binary batch codes with
  improved redundancy,'' \emph{IEEE Transactions on Information Theory},
  vol.~66, no.~12, pp. 7360--7370, 2020.

\bibitem{rawat2016batch}
A.~S. Rawat, Z.~Song, A.~G. Dimakis, and A.~G{\'a}l, ``Batch codes through
  dense graphs without short cycles,'' \emph{IEEE Trans. Inf. Theory}, vol.~62,
  no.~4, pp. 1592--1604, 2016.

\bibitem{woodruff2012quadratic}
D.~P. Woodruff, ``A quadratic lower bound for three-query linear locally
  decodable codes over any field,'' \emph{J. Computer Sci. Technol.}, vol.~27,
  no.~4, pp. 678--686, 2012.

\bibitem{vardy2016constructions}
A.~Vardy and E.~Yaakobi, ``Constructions of batch codes with near-optimal
  redundancy,'' in \emph{Proc. IEEE Int. Symp. Inf. Theory (ISIT)}, 2016, pp.
  1197--1201.

\bibitem{li2021improved}
R.~Li and M.~Wootters, ``Improved batch code lower bounds,'' \emph{arXiv
  preprint arXiv:2106.02163}, 2021.

\bibitem{ben2006LCC}
E.~Ben-Sasson and M.~Sudan, ``Robust locally testable codes and products of
  codes,'' \emph{Random Structures Algorithms}, vol.~28, no.~4, pp. 387--402,
  2006.

\bibitem{viderman2015combination}
M.~Viderman, ``A combination of testability and decodability by tensor
  products,'' \emph{Random Structures Algorithms}, vol.~46, no.~3, pp.
  572--598, 2015.

\bibitem{kopparty2017high}
S.~Kopparty, O.~Meir, N.~Ron-Zewi, and S.~Saraf, ``High-rate locally
  correctable and locally testable codes with sub-polynomial query
  complexity,'' \emph{Journal of the ACM (JACM)}, vol.~64, no.~2, pp. 1--42,
  2017.

\bibitem{trevisan2004some}
L.~Trevisan, ``Some applications of coding theory in computational
  complexity,'' in \emph{Electron. Colloq. Comput. Complex. (ECCC)}, 2004.

\bibitem{kopparty2017local}
S.~Kopparty and S.~Saraf, ``Local testing and decoding of high-rate
  error-correcting codes,'' in \emph{Proc. Electron. Colloq. Comput. Complex.
  (ECCC)}, vol.~24, 2017, p. 126.

\bibitem{erdos1964extremal}
P.~Erd{\"o}s, ``On extremal problems of graphs and generalized graphs,''
  \emph{Israel J. Math.}, vol.~2, no.~3, pp. 183--190, 1964.

\bibitem{sudan2001ideal}
M.~Sudan, ``Ideal error-correcting codes: Unifying algebraic and
  number-theoretic algorithms,'' in \emph{Proc. Int. Symp. Applied Algebra
  Algebr. Algor. Error-Correcting Codes (AAECC)}.\hskip 1em plus 0.5em minus
  0.4em\relax Springer, 2001, pp. 36--45.

\bibitem{demillo1977probabilistic}
R.~A. DeMillo and R.~J. Lipton, ``A probabilistic remark on algebraic program
  testing,'' \emph{Inf. Process. Lett.}, vol.~7, no.~4, p. 193–195, 1977.

\bibitem{zippel1979probabilistic}
R.~Zippel, ``Probabilistic algorithms for sparse polynomials,'' in \emph{Proc.
  Int. Symp. Symb. Algebr. Manipul. (SYMSAC)}.\hskip 1em plus 0.5em minus
  0.4em\relax Springer, 1979, pp. 216--226.

\end{thebibliography}

\appendices

\section{Proof of Proposition~\ref{prop::code cardinality}}\label{ss::injective map}
The proof is twofold, we need to show that
\begin{enumerate}[wide=\parindent]%
    \item[\textit{(Distinction)}] the evaluation of every monomial $\X^{\d}$ with $\deg_q(\d)\leq s-1$, which we refer to as a \textit{type-$s$} monomial, gives a unique word
    \item[\textit{(Inclusion)}] these words are contained in the $[m,s,d,q]$ lifted multiplicity code as in Defintion~\ref{def::lifted mult code} .
\end{enumerate}

To show that the words are distinct, it is sufficient to prove that for an arbitrary non-trivial linear combination, written as $f(\X)$, of type-$s$ monomials, its evaluation is not equal to the all-zero codeword. Our proof is a straightforward generalization of~\cite[Lemma 14]{li2019lifted}.

We prove the proposition by induction on $m$ and $s$. More precisely, we deduce the statement for $(m, s)$ from the cases for $(m-1, s)$ and $(m, s-1)$. The base case $m=1$ is equivalent to~\cite[Lemma~11]{li2019lifted}. In the base case $s=1$ the degree of each variable in $f$ is at most $q-1$. Then the proposition follows from DeMillo–Lipton–Zippel Theorem~\cite{demillo1977probabilistic,zippel1979probabilistic}, which states that such polynomial can't have more than $q^m-(q-(q-1))^m=q^m-1$ zeroes.

Now we prove the inductive step. Assume that $f(\X)$ is a non-trivial linear combination of type-$s$ monomials such that $f(\X)\equiv_s 0$. Consider the polynomial $g(X_1,\ldots,X_{m-1}):=f(X_1,\ldots, X_{m-1}, c)$ in $m-1$ variables, where $c\in \F_q$ is fixed. By the inductive hypothesis, we conclude that $g\equiv_s 0$. Hence, $(X_m-c)$ divides $f(\X)$ for all $c\in \F_q$, so $(X_m^q-X_m)$ divides $f(\X)$. Therefore, $f(\X)$ can be represented as $f(\X)=(X_m^q-X_m)g(\X)$.

It is easy to see that $g(\X)$ is a linear span of type-$(s-1)$ monomials. Taking the $\i$th derivative of $f(\X)$ for any $\i\in\Z_{\ge}^m$ with $i_m\geq 1$ we obtain
\begin{align*}
  f^{(\i)}(\X)=(X_m^q-X_m)g^{(\i)}(\X)-g^{(\j)}(\X),
\end{align*}
where $\j=(i_1,\ldots, i_{m-1}, i_m-1)$. The left-hand side is equal to zero for all $\x\in \F_q^m$ and $\i\in\Z_{\ge}^m$ with $\deg(\i)\le s-1$. The right-hand side equals to $-g^{(\j)}(\x)$ for all $\x\in \F_q^m$ and all $\j\in\Z_{\ge}^m$ with $\sum\limits_{l=1}^{m-1} j_l<s-1$. By the induction hypothesis $g(\X)$ is the zero polynomial, thus, $f(\X)$ is the zero polynomial as well. This concludes the proof of the distinction property.

To show the inclusion, we prove that every $(d,s)^*$-good monomial $f(\X)=\X^\d$ over $\F_q$ satisfies the property that for any line $L\in\L_m$, the restriction $f|_L$ is equivalent up to order $s$ to an univariate polynomial of degree less than $d$. Let a line $L$ be parameterized as $(\w + \v T)|_{T\in\F_q}$ and $\0$ be the all-zero vector. Then, we have that
\begin{align*}
f|_L &=(\w+\v T)^\d \\
&= \sum_{\0\le \i\le \d}\prod_{j=1}^{m}v_j^{i_j}w_j^{d_j-i_j}\binom{d_j}{i_j} T^{i_j}\\
&\equiv_s \sum_{k=0}^{qs-1} c_k T^k := f^*(T),
\end{align*}
where $c_k$ denotes the coefficients of the unique polynomial of degree $\leq qs-1$ that is equivalent to $f|_L$ (cf. Proposition~\ref{pr::reducing the power}). Recall that $s$ and $q$ are powers of $2$. Hence, we have $f|_L(T) = f^*(T) \pmod{T^{qs} + T^s}$ by Proposition~\ref{pr::reducing the power}, so the coefficients $[T^s]f|_L$ that contribute to the coefficient $c_k$ are exactly those for which $s=\deg(\i) \Modsp{q}{s} = k$, and we obtain
\begin{equation}\label{eq::coefficients equivalent polynomial}
c_k\eqdef \sum_{\substack{\0\le \i\le \d \\ \deg(\i)\Modsp{q}{s} = k }}\prod_{j=1}^{m}v_j^{i_j}w_j^{d_j-i_j}\binom{d_j}{i_j}.
\end{equation}
By Definition~\ref{def::bad (d^*,s) monomial}, for $k\ge d$, there is no $\i\in\Z_{qs}^m$ such that $\i\le_2\d$ and $\deg(\i)\Modsp{q}{s} = k$. Thus, for $k\ge d$ and every $\i$ used in the summation of \eqref{eq::coefficients equivalent polynomial}, there exists some coordinate $j\in[m]$ such that $i_j\not\le_2 d_j$. By Lucas's Theorem (e.g., see~\cite{guo2013new,li2019lifted}), for integers $d_j=(d_j^{(\ell-1)},...,d_j^{(0)})_2$ and $i_j=(i_j^{(\ell-1)},...,i_j^{(0)})_2$ it holds that
\begin{equation*}
    \binom{d_j}{i_j} = \prod_{\xi=0}^{\ell-1} \binom{d_j^{(\xi)}}{i_j^{(\xi)}} \mod 2  .
\end{equation*}
It follows that if $i_j\not\le_2 d_j$ the coefficient $\binom{d_j}{i_j}=0$ in $\F_q$ (as $q$ is a power of two) and therefore $c_k=0$ for all $k\ge d$.

We have proved that the restriction of $\X^\d$ to any line is an univariate polynomial of degree at most $d-1$. Therefore, the $[m,s,d,q]$ lifted multiplicity code includes the codewords
\begin{equation*}
\{(\a^\d)|_{\a\in\F_q^m}\ : \ \X^\d \in \mathcal{F}_{q}(m,s,d) \} \ . %
\end{equation*}
The inclusion of their linear combinations over $\F_q$ follows trivially from the proof.
\qed

\section{Lifted multiplicity code and lifted multiplicity monomial code}\label{ss::equivalenceLiftedRS}

We now give an example showing that lifted multiplicity codes are not necessarily spanned by the set of good monomials. %
 Let $d=qs-2$, $s=2$, and $q>2$.
 Denote by $M_1(\X)$ and $M_2(\X)$ the monomials
\begin{align*}
  M_1(\X) &:= \X^{\d^{(1)}} = X_1^{qs-2}X_2\\
  M_2(\X) &:= \X^{\d^{(2)}} = X_1^{(s-1)q-1}X_2^q\ ,
\end{align*}
so $d_{1}^{(i)} = qs-2$, $d_{2}^{(1)} = 1$, $d_{1}^{(2)}=(s-1)q-1$, and $d_{2}^{(2)}=q$.
Both monomials are type-$s$ as
\begin{align*}
 \deg_q(\d^{(1)}) = \deg_q(\d^{(2)}) = qs-1 < 2(q-1) + (s-1)q
\end{align*}
Further, both are $(d,s)^*$-bad, as the vectors $\i^{(1)} = \d^{(1)}$ and $\i^{(2)} = \d^{(2)}$ fulfill Definition~\ref{def::bad (d^*,s) monomial} for each monomial, respectively. Also, their evaluation is not contained in an $[m,s,d,q]$ lifted multiplicity code, since for the line $(0,w_2)+(1,v_2)T \in \L_2$ we have
\begin{align*}
  [T^{qs-1}] M_1(T,w_2+v_2 T) &= v_2\\
  [T^{qs-1}] M_2(T,w_2+v_2 T) &= v_2^q \ .
\end{align*}
However, the evaluation of their sum, i.e., the polynomial
\begin{align*}
  P(\X) := M_{1}(\X) + M_{2}(\X) \ ,
\end{align*}
is contained in the $[m,s,d,q]$ lifted multiplicity code as
\begin{align*}
  [T^{qs-1}] &P(w_1 + v_1T,w_2+v_2 T) \\
             &= [T^{qs-1}]  M_1(w_1 + v_1T,w_2+v_2 T)\\
  &\qquad + [T^{qs-1}]  M_2(w_1 + v_1 T,w_2+v_2 T) \\
  &= v_1^{qs-2}v_2 + \underbrace{v_1^{(s-1)q-1}v_2^q}_{\stackrel{\mathsf{(a)}}{=} v_1^{qs-2}v_2} = 0 \ ,
\end{align*}
where $\mathsf{(a)}$ holds because $v_1, v_2 \in \F_q$.

\section{Examples of rate improvements through lifting}

To provide some intuition and show how lifting can improve the rate of lifted RS codes and lifted multiplicity codes, we provide examples for fixed sets of parameters.

\subsection{RM codes vs. lifted RS codes} \label{ss::improvement RM vs LRS}
Let $f(X_1,X_2)=X_1^2 X_2^2$. Then the $[2,3,4]$ lifted RS code includes the codeword $\c=(f(a_1,a_2))|_{(a_1,a_2)\in\F_4^2}$ as for every line $L$, the degree of $f|_L$ is at most $2<3=d$. Indeed, given a line $L$ parameterized as $(w_1 + v_1 T , w_2 + v_2 T)|_{T\in \F_4}$ in $\F_4^2$, we have
\begin{align*}
f|_L&=f(v_1 T +w_1, v_2 T +w_2)=(v_1 T +w_1)^2 (v_2 T +w_2)^2\\
&\overset{(i)}{=}(v_1^2 T^2 +w_1^2)(v_2^2 T^2 +w_2^2)\\
&\overset{(ii)}{=} (v_1^2w_2^2+v_2^2w_1^2)T^2
+v_1^2v_2^2T+w_1^2 w_2^2,
\end{align*}
where in $(i)$ we used the property $2v = 0$ for any $v\in\F_4$, and $(ii)$ is implied by the fact that $T^4=T$ in $\F_4[T]$. On the other hand, the $2$-variate RM code of order $3$ doesn't contain $\c$ as the degree of $f$ is $4$, which is larger than $3$.

\subsection{Multiplicity codes vs. lifted multiplicity codes} \label{ss::improvement LMC vs MC}

Let $m=s=2$, $q=4$, and $d=qs-1=7$. Consider the monomial $M(\X) := X_1^2 X_2^6$. The degree of this monomial is $\deg(M(\X))=8>d$, so its evaluation is not contained in the $[2,2,7,4]$ multiplicity code, as it only contains evaluations of degree $<d$ polynomials.

By Definition~\ref{def::lifted mult code}, the evaluation of $M(\X)$ is contained in the $[2,2,7,4]$ lifted multiplicity code if for every line $L\in \mathcal{L}_m$ there exists a polynomial $g(T)\in \mathcal{F}_{q}(d)$ such that the restriction of $M(\X)$ to $L$ is equivalent to $g(T)$. First, note that $M(\X)$ is a type-$s$ monomial, as $\deg_q(M(\X)) = 1 \leq s-1$. Its evaluation in an arbitrary line $L\in\mathcal{L}_2$ is given by
\begin{align*}
    &M(\X)|_L \\
    &= (w_1 + v_1T)^2 (w_2+v_2T)^6\\
    &= (w_1^2 + v_1^2T^2)(w_2^6+w_2^4v_2^2 T^2 + w_2^2v_2^4 T^4 + v_2^6T^6) \\
     &= w_1^2w_2^6 + (w_1^2w_2^4v_2^2 + v_1^2w_2^6)T^2 + (w_1^2w_2^2v_2^4 + v_1^2w_2^4v_2^2)T^4 \\
   &\qquad + (w_1^2v_2^6+v_1^2 w_2^2v_2^4)T^6 + v_1^2v_2^6T^8 .
\end{align*}
By Proposition~\ref{pr::reducing the power} and because $s$ and $q$ are powers of $2$, we know that there exists an equivalent polynomial $M^*(T)$ of degree at most $qs-1=7$ such that $M(\X)|_L \equiv_s M^*(T) \pmod{T^{8}+T^2}$. Here, we obtain this polynomial by substracting $v_1^2v_2^6(T^{8}+T^2)$ from $M(\X)|_L$, which gives
\begin{align*}
  M^*(T) &=w_1^2w_2^6 + (w_1^2w_2^4v_2^2 + v_1^2w_2^6+v_1^2v_2^6)T^2 \\
  &\quad + (w_1^2w_2^2v_2^4 + v_1^2w_2^4v_2^2)T^4 + (w_1^2v_2^6+v_1^2 w_2^2v_2^4)T^6 .
\end{align*}
As the degree of this polynomial is $\deg(M^*(T))<d=7$ its evaluation is contained in the $[2,2,7,4]$ lifted multiplicity code, thereby increasing its dimension compared to the $[2,2,7,4]$ multiplicity code.

\section{Comparison of New Bounds to Known Results} \label{app:comparison}

In Tables~\ref{tab:comparisonNonBinaryPIR} we summarize the ranges of $\epsilon$ in which each bound on the required redundancy of $n^{\epsilon}$-PIR and $n^{\epsilon}$-batch codes of dimension $n$ is best among the known results.
\begin{table*}[h]
  \renewcommand\arraystretch{1.3}
  \centering
  \caption{Non-binary PIR codes.}
  \begin{tabular}{cccc}
    Given in &Reference& Based on & Best for $\epsilon$ in \\ \hline
    Lemma~\ref{lem::known1}, Item \ref{item:asi}) & \cite{asi2018nearly}& multiplicity codes & -- \\
    Lemma~\ref{lem::known1}, Item \ref{item:FGW}) & \cite{FGW17} & partially lifted codes & -- \\
    Lemma~\ref{lem::known1}, Item \ref{item:LiLMC}) & \cite{li2019lifted}& lifted mult. codes with $m=2$& $(0,\frac{1}{2}]$ \\
    Lemma~\ref{lem::known1}, Item \ref{item:GuoQaryPIR}) & \cite{guo2013new}& lifted RS codes& -- \\
    Lemma~\ref{lem::known1}, Item \ref{item:LRS}) & \cite{polyanskii2019lifted}& lifted RS codes with $m=3$& $\{2/3\}$ \\ \hdashline
    Theorem~\ref{th::asymptotic non-binary disjoint repair group code} & This work & lifted mult. codes & $(0,1)$
  \end{tabular}
  \label{tab:comparisonNonBinaryPIR}
\end{table*}

\begin{table*}[h]
  \renewcommand\arraystretch{1.3}
  \centering
  \caption{Binary PIR codes.}
  \begin{tabular}{cccc}
    Given in &Reference& Based on & Best for $\epsilon$ in \\ \hline
    Lemma~\ref{lem::known2}, Item \ref{item:binPIRLRS})&\cite{guo2013new} & lifted RS codes & $\{0.5\}$ \\
    Lemma~\ref{lem::known2}, Item \ref{item:binPIRasi})&\cite{LC04, asi2018nearly}&array and one-step majority logic dec. codes & $[0,0.273]\setminus \bigcup_{a\in \mathbb{N}}[\frac{\log(2-2^{-a})}{2a}, \frac{1}{2a}]$ \\
    Lemma~\ref{lem::known2}, Item \ref{item:binPIRmult})&\cite{asi2018nearly}&binary image of mult. codes & -- \\
    Lemma~\ref{lem::known2}, Item \ref{item:binPIRLMC2})&\cite{li2019lifted}&binary image of lifted mult. codes with $m=2$ & $[0.273,0.5)$ \\
    Lemma~\ref{lem::known2}, Item \ref{item:binPIRLRS3})&\cite{polyanskii2019lifted}&binary image of lifted RS codes with $m=3$ & $\{2/3\}$ \\
    Lemma~\ref{lem::known2}, Item \ref{item:binPIRwedge})& \cite{hastings2020wedge} & wedge-lifted codes & $  [0,0.273]\bigcap \left\{\bigcup_{a\in \mathbb{N}}[\frac{\log(2-2^{-a})}{2a}, \frac{1}{2a}]\right\}$\\ \hdashline
    Theorem~\ref{th::binary disjoint repair group property code} & This work & binary image of lifted mult. codes & $[0.273,1)$
  \end{tabular}
  \label{tab:comparisonBinaryPIR}
\end{table*}

\begin{table*}[h]
  \renewcommand\arraystretch{1.3}
  \centering
  \caption{Non-binary batch codes.}
  \begin{tabular}{cccc}
    Given in &Reference& Based on & Best for $\epsilon$ in \\ \hline
    Lemma~\ref{lem::knownBatch non-binary}, Item \ref{item:qaryBatchLMC2}) &  \cite{polyanskaya2020binary,asi2018nearly}&lift. mult. codes and mult. codes with $m=2$ & $[0.25,0.432]$ \\
    Lemma~\ref{lem::knownBatch non-binary}, Item \ref{item:qaryBatchMultOld}) & \cite{asi2018nearly}& PIR codes (mult. codes with $m\ge 3)$ & -- \\
    Lemma~\ref{lem::knownBatch non-binary}, Item \ref{item:qaryBatchFG}) & \cite{polyanskaya2020binary}&finite geometry design & $(0,0.25]$\\ \hdashline
    Theorem~\ref{th::asymptotic non-binary batch code} & This work & lifted RS codes & $[0.432,0.582]$ \\
    Theorem~\ref{th::asymptotic non-binary batch code from PIR codes} & This work & PIR codes (lifted mult. codes) & $[0.582,1)$
  \end{tabular}
  \label{tab:comparisonNonBinaryBatch}
\end{table*}

\begin{table*}[h]
  \renewcommand\arraystretch{1.3}
  \centering
  \caption{Binary batch codes.}
  \begin{tabular}{cccc}
    Given in &Reference& Based on & Best for $\epsilon$ in \\ \hline
    Lemma~\ref{lem::knownBatch}, Item \ref{item:binBatchLMC2}) & \cite{polyanskaya2020binary} & binary image of  lifted mult. codes with $m=2$ & $[0.269,0.41]$ \\
    Lemma~\ref{lem::knownBatch}, Item \ref{item:binBatchMultOld}) & \cite{asi2018nearly} &  PIR codes (binary image of mult. codes with $m\ge 3$) & -- \\
    Lemma~\ref{lem::knownBatch}, Item \ref{item:binBatchFG}) & \cite{polyanskaya2020binary} & finite geometry design &  $(0,0.269]$ \\ \hdashline
    Theorem~\ref{th::binary batch code} & This work & binary image of lifted RS codes & $[0.41,0.648]$ \\
    Theorem~\ref{th::binary batch code from PIR codes} & This work & PIR codes (binary image of lifted mult. codes) & $[0.648,1)$
  \end{tabular}
  \label{tab:comparisonBinaryBatch}
\end{table*}

\end{document}